\documentclass[a4paper]{llncs}

\usepackage{wrapfig}
\usepackage{listings}
\usepackage{mathtools}
\usepackage{bm}
\usepackage{amssymb}
\usepackage{thm-restate}

\usepackage{subcaption}
\usepackage{pslatex}

\usepackage{paralist}
\usepackage[shortlabels]{enumitem}

\usepackage[noend]{algpseudocode}
\usepackage{algorithm}
\usepackage{url}
\usepackage{xspace}
\usepackage[capitalise,noabbrev]{cleveref}
\usepackage{todonotes}

\algrenewcommand\algorithmicdo{}

\usepackage{tikz}
\usetikzlibrary{
  arrows.meta,
  automata,
  positioning,
  fit,
  shapes,
  decorations.pathmorphing,
}
\tikzset{every state/.style={minimum size=0pt}}
\tikzset{mstate/.style={draw,rectangle,rounded corners,fill=blue!0,inner xsep=.3em,inner ysep=0em,text height=2ex,text depth=.9ex}}

\usepackage{braket}
\let\setbuilder\set
\newcommand{\simpleset}[1]{\{{#1}\}}
\usepackage{xstring}
\renewcommand{\set}[1]{\normalexpandarg\IfSubStr{#1}{|}{\setbuilder{#1}}{\simpleset{#1}}}
\usepackage{trimclip}

\makeatletter
\DeclareRobustCommand{\shortto}{%
  \mathrel{\mathpalette\short@to\relax}%
}

\DeclareRobustCommand{\shortminus}{%
  \mathrel{\mathpalette\short@minus\relax}%
}

\newcommand{\short@to}[2]{%
  \mkern2mu
  \clipbox{{.5\width} 0 0 0}{$\m@th#1\vphantom{+}{\rightarrow}$}%
}

\newcommand{\short@minus}[2]{%
  \mkern2mu
  \clipbox{{.5\width} 0 0 0}{$\m@th#1\vphantom{+}{-}$}%
}
\makeatother

\newcommand{\labeledto}[1]{{{\shortminus}\hspace{-2pt}\raisebox{0.16ex}{$\scriptstyle\{ #1\hspace{-0.28pt}\}$}\hspace{-2.2pt}{\shortto}}}

\newcommand{\scriptlabeledto}[1]{{{\shortminus}\hspace{-1.0pt}\raisebox{0.12ex}{$\scriptscriptstyle\{ #1\hspace{-0.28pt}\}$}\hspace{-1.6pt}{\shortto}}}

\newcommand\move[3]{
\mathchoice
{#1\,\labeledto{#2}\,#3}
{#1\labeledto{#2}#3}
{#1\scriptlabeledto{#2}#3}
{#1\scriptlabeledto{#2}#3}
}

\renewcommand{\phi}{\varphi}

\newcommand{\mem}{\mathfrak{m}}
\newcommand{\smem}{\mathfrak{s}}

\newcommand{\nat}{\mathbb{N}}

\newcommand{\PowerSetFin}[1]{\mathcal{P}_{\text{fin}}(#1)}
\newcommand{\PowerSet}[1]{\mathcal{P}(#1)}

\newcommand{\filter}{\nabla}
\newcommand{\filterof}[1]{\nabla_{\!\!#1}}
\newcommand{\fancyfilter}[1]{F_{c < \mxof{c}}}

\newcommand{\Minterms}[1]{\textit{Minterms}(#1)}

\newcommand{\nfaof}[1]{\mathtt{NFA}(#1)}
\newcommand{\confaut}[1]{\mathtt{Conf}(#1)}
\newcommand{\dfaof}[1]{\mathtt{DA}(#1)}
\newcommand{\scaut}[1]{\mathtt{DCSA}(#1)}
\newcommand{\caof}[1]{\mathtt{CA}(#1)}
\newcommand{\csaof}[1]{\mathtt{CSA}(#1)}

\newcommand{\ii}[1]{#1^{\mathtt{a}}}
\newcommand{\csaiiof}[1]{\ii{\mathtt{DCSA}}(#1)}
\newcommand{\pow}{\PowerSet}

\newcommand{\bigo}[1]{O(#1)}

\newcommand{\regex}[1]{\texttt{#1}}
\newcommand{\rexeps}{\epsilon}
\newcommand{\rexa}{\regex{a}}
\newcommand{\rexuni}[2]{#1\regex{|}#2}
\newcommand{\rexconcat}[2]{#1\,#2}
\newcommand{\rexstar}[1]{#1\regex{*}}
\newcommand{\rexcount}[2]{#1\regex{\{}#2\regex{\}}}
\newcommand{\rexbracket}[1]{\regex{(}#1\regex{)}}

\newcommand{\boundof}[1]{\max_{#1}}
\newcommand{\cntof}[1]{\mathit{cnt}_{#1}}
\newcommand{\concat}{\cdot}
\newcommand{\expof}[1]{(\boundof{#1})^{\cntof{#1}}}
\newcommand{\olof}[1]{\mathit{ol}(#1)}
\newcommand{\ol}{\mathit{l}}
\newcommand{\conf}{\alpha}
\newcommand{\merge}{\mathit{merge}}
\newcommand{\cost}{\mathit{cost}}
\newcommand{\encode}{\mathit{enc}}
\newcommand{\decode}{\mathit{dec}}
\newcommand{\plusmark}{\texttt{\string+}}
\newcommand{\marked}[1]{#1^{\plusmark}}

\newcommand\conflict{\mathit{Conflict}}

\newcommand{\trin}{\Delta^{\mathit{in}}_x}
\newcommand{\trout}{\Delta^{\mathit{out}}_x}
\newcommand{\trinner}{\Delta^{\mathit{inner}}_x}
\newcommand{\trinc}{\Delta^{+}_x}
\newcommand{\trrestart}{\Delta^1_x}

\newcommand{\indexby}[2]{{\mathsf{index}_{#2}(#1)}}

\newcommand{\mxof}[1]{\max_#1}
\newcommand{\mnof}[1]{\min_#1}

\newcommand{\markedregex}[1]{\overline{#1}}
\newcommand{\first}[1]{\text{first}({#1})}
\newcommand{\last}[1]{\text{last}({#1})}
\newcommand{\follow}[1]{\text{follow}({#1})}
\newcommand{\emptyctr}{\text{null}}
\newcommand{\subctrs}[1]{\text{counters}(#1)}
\newcommand{\lowerctr}[1]{\text{lower}(#1)}
\newcommand{\upperctr}[1]{\text{upper}(#1)}
\newcommand{\entry}[1]{\first{#1}}

\newcommand{\states}[1]{Q_#1}
\newcommand{\head}[1]{q_#1}

\newcommand{\ass}{\mathbin{:=}}

\newcommand{\finc}[1]{#1 \oplus 1}

\newcommand\fol{\Gamma}
\newcommand\pset{P}
\newcommand\fset{F}
\newcommand\vset{X}
\newcommand\vass{\nu}
\newcommand\tf{\{0,1\}}
\newcommand\formulasof[2]{\mathit{\mathsf{QFF}}_{#1,#2}}

\newcommand\termsof[2]{\mathit{\mathsf{Terms}}_{#1,#2}}

\newcommand\interp{\mathbb{I}}

\newcommand\interpof[1]{#1^\interp}
\newcommand\jinterp{\interp_{\mathtt{set}}}
\newcommand\jinterpof[1]{#1^{\jinterp}}
\newcommand\cinterp{\interp_{\mathtt{cnt}}}
\newcommand\cinterpof[1]{#1^{\cinterp}}

\newcommand\clang{\fol_{\mathtt{cnt}}}
\newcommand\slang{\fol_{\mathtt{set}}}
\newcommand\sharedu{u^{\mathtt{sh}}}
\newcommand\lval{\mathtt{lval}}
\newcommand\qq {{q^\circ}}

\newcommand\ac{\mathit{Act}}
\renewcommand\sc[1]{#1^{\textup{\ensuremath{\scriptstyle{\texttt{\string{\hspace{-1.3pt}\string}}}}}}}

\title{Fast Matching of Regular Patterns with Synchronizing Counting (Technical Report)\thanks{To appear in FoSSaCS'23.}}

\author{
Luk\'{a}\v{s} Hol\'{i}k
\and 
Juraj S\'{i}\v{c}
\and 
Lenka Turo\v nov\'a
\and
Tom\'{a}\v{s} Vojnar
}

\institute{Brno University of Technology, Czech Republic\\
\email{\{holik,sicjuraj,ituronova,vojnar\}@fit.vut.cz}}

\begin{document}
\maketitle

\begin{abstract}
Fast matching of regular expressions with \emph{bounded repetition}, aka
\emph{counting}, such as $\regex{(ab)\{50,100\}}$, i.e., matching linear in the
length of the text and independent of the repetition bounds,
has been an open problem for at least two decades.  
We show that, for a wide class of regular expressions with counting, which we call
\emph{synchronizing}, fast matching is possible.
We empirically show that the class covers nearly all counting used in usual applications of regex matching.
This complexity result is based on an improvement and analysis of a recent matching algorithm that compiles regexes to deterministic counting-set automata (automata with registers that hold sets of numbers).
\end{abstract}

\section{Introduction}
Fast matching of regular expressions with \emph{bounded repetition}, aka
\emph{counting}, has been an open problem for at least two decades (cf., e.g., \cite{Sperberg-McQueen-ExtendedFAWeb}). 
The time complexity of the standard matching algorithms run on a regex such as $\regex{.*a.\{100\}}$ is, at best, dominated by the \emph{length of the text multiplied by the repetition bounds}. This makes matching prone to unacceptable slowdowns since the length of the text as well as the repetition bounds are often large.
In this paper, we provide a theoretical basis for matching of bounded repetition with a much more reliable performance. 
We show that a large and practical class of regexes with counting theoretically
allows {\bf fast matching}---in time {\bf independent of the counter bounds} and {\bf linear in the length of the text}. 


The problem also has a strong practical motivation.
Regex matching is 
used for searching, data
validation, detection of information leakage, parsing, replacing,
data scraping, syntax highlighting, etc.
It is natively supported in most programming languages~\cite{regexwiki}, and ubiquitous 
(used in 30--40\,\% of Java, JavaScript, and Python software \cite{rethinkingregexes,DBLP:conf/sigsoft/WangS18,DBLP:conf/sigsoft/DavisCSL18,DBLP:conf/issta/ChapmanS16}). 
%
Efficiency and predictability of regex matching is important.  
An extreme run-time of matching 
can have serious consequences, such as a~failed input validation against
injection attacks \cite{ADMIN2020}
and events like the outage of Cloudflare
services~\cite{cloudfareoutage}.  
Regexes vulnerabilities are also a doorway for the \emph{ReDoS (regular expression denial of service) attack}, 
in which the attacker crafts a text to overwhelm a matcher 
(as, e.g., in the case of the outage of
StackOverflow~\cite{stackoutage} or the websites exposed due to their use of
the popular Express.js framework~\cite{expressjsoutage}).
%
%
ReDoS has been widely recognized as a common and serious threat \cite{rethinkingregexes,redosimpact,recentjamie},
\mbox{with counting in regexes begin especially dangerous \cite{usenix}.}

\paragraph{Matching algorithms and complexity.}
The potential instability of the pattern matchers is in line with the worst-case complexity of the matching algorithms.
The most widely used approach to matching is backtracking
(used, e.g., in standard matchers of .NET, Python, Perl, PHP, Java, JavaScript, Ruby)
for its simplicity and ease of implementation of advanced features such as
back-references or look-arounds. It is, however, at worst exponential to the length of the matched text and prone to ReDoS. 
Even though this can be improved, for instance by memoization~\cite{recentjamie},
the fastest matchers used in performance critical applications all use automata-based algorithms instead of backtracking.
The basis of these approaches is Thompson's algorithm \cite{thompsonmatching} (also referred to as
\emph{online NFA-simulation}). Together with many optimizations, it is implemented in Intel's Hyperscan \cite{hyperPaper}.  
When combined with caching, it becomes the on-the-fly subset construction of a
DFA, also called \emph{online DFA-simulation}
(implemented in RE2 from Google, GNU grep, SRM, or the standard matcher of Rust~\cite{re2,grep,VSXW19,rust}).
Without counting, the major factor in the worst-case complexity is 
$O(n m^2)$, 
with $n$ being the length of the text and $m$ the size of the number of character occurrences in the regex ($m$ is smaller than size of the regex, the length of string defining it).
We say that the \emph{character cost}, i.e., the cost of extending the text with one character, is $m^2$.
This is the cost of iterating through transitions of an NFA with $O(m)$ states and $O(m^2)$ transitions compiled from the regex by some classical construction \cite{Antimirov96_partderivatives,Glu61,hro97}.  

Extending the syntax of regexes with \emph{bounded quantifiers} (or
\emph{counters}), such as
$\regex{(ab)\{50,100\}}$, increases the character complexity dramatically.  
Given $k$ counters with the maximum bound $\ell$, 
the number of NFA states rises to $O(m\ell^k )$, the number of transitions as well as the character cost to $O((m\ell^k)^2)$. 
%
For instance, the minimal DFA for $\regex{.*a.\{k\}}$ (i.e., $a$ appears $k$
characters from the end) has more than $2^k$ states.
Moreover, note that, since $k$ is written as a~decadic numeral, its value is exponential in the size of the regex. 
This makes matching with already moderately high $k$ prone to significant slowdowns and ReDoS vulnerabilities  with virtually every mainstream matcher (see
\cite{oopsla,usenix}). 
At the same time, repetition bounds easily reach thousands, in extreme tens of millions (in real-life XML \cite{cikm15}). 
Writing a~dangerous counting expression is easy and it is hard to identify.
%
%
Security-critical solutions may be vulnerable to counting-related ReDoS \cite{usenix} despite an extra effort spent in regex design and testing,
hence developers sometimes avoid counting, use workarounds and restrict functionality. 

The problem of matching with bounded repetition 
%
has been addressed from the theoretical as well as from the practical perspective by a number of authors \cite{GGM12,cikm15,Hovland09,KilTuh07,SEJ08,aplas19,KT03,oopsla}.
%
%
From these, the recent work \cite{oopsla} is the only one 
offering fast matching 
for a practically significant class of regexes.
%
%
The algorithm of \cite{oopsla} compiles a regex with counting to a non-deterministic \emph{counting automaton (CA)}, 
an automaton with counters that can be incremented, reset, and compared with a constant. 
%
The crux of the problem is then to convert the CA to a succinct deterministic machine that could be simulated fast in matching.  
The work \cite{oopsla} achieves this by determinizing the CA into a \emph{counting-set automaton (CSA)}, 
an automaton with registers that hold \emph{sets} of numbers. 
%
Its size is independent of the counter bounds and it updates the sets by a handful of operations that are all constant time, regardless the size of the sets. 
%
%
%
However, regexes outside the supported class do appear, the class has no syntactic characterization, and it is hard to recognize (as demonstrated also by an incorrect proposal of a syntactic class in \cite{oopsla} itself, see \cref{app:counterexample}).
%
%
For instance, $\regex{.*a\{5\}}$ or $\regex{(ab)\{5\}}$ are handled, but $\regex{.*(aa)\{5\}}$ or  $\regex{.*(ab)\{5\}}$ are not (the requirement is technical, see \cref{sec:determinization}).
%
%
%
%
%
%
%
%
%

\paragraph{Our contribution.}
In this paper, we

{\bf \begin{enumerate}
\item  generalize the algorithm of \cite{oopsla} to extend the class of handled regexes and
\item derive a useful syntactic characterization of the extended class.
\end{enumerate}}
The derived class is characterized by \emph{flat counting} (counting operators are not nested) where repetitions of each counted expression
$R$ are \emph{synchronizing} (a word
from $R^n$ cannot have a prefix from $R^{n+1}$).
It is the first clearly delimited practical class of regexes with counting that allows fast matching.
It includes the easily recognizable and frequent case 
where every word in $R$ has exactly one occurrence of a \emph{marker},
a letter or a word from a finite set of markers that unambiguously identifies
each occurrence of $R$ (note that even this simple class was not handled by any previous fast algorithms, including \cite{oopsla}). 
%
%
%
In a our experiment with a large set of regexes from various sources,
99.6\,\% of non-trivial flat counting was synchronizing and 99.2\,\% was letter-marked. 


To obtain the results (1) and (2) above, {\bf we first modify the
determinization of \cite{oopsla} to include the entire class of regexes with flat counting}.
In a nutshell, this is achieved by two changes:  
(i) We allow copying and uniting of sets stored in registers, and 
(ii) in the determinization, we index counters of the CA by its states
to handle CA in which nondeterministic runs that reach different states reach different counter values.

These modifications come with the main technical challenge that we solve in this paper: 
copying and uniting sets is not constant-time but linear to the size of the sets.
This would make the character cost linear in the counter bound $\ell$ again. 
%
%
To remove the dependency on the counter bounds, 
we augment the determinization by optimizations that avoid the copying and uniting. 
First, to alleviate the cost of uniting, we store intersections of sets stored in registers in new shared registers, 
so that the intersection does not contribute to the cost of uniting the registers.   
Then, to increase the impact of intersection sharing, 
we synchronize register updates in order to make their intersections larger. 
We then show that if the CSA \emph{does not replicate registers}, i.e, each register can in a transition appear on the right-hand side of only one register assignment, then it never copies registers and the cost of unions can be amortised. 
Finally, {\bf we define the class of regexes with \emph{synchronizing counting} for which the optimized CsA do not replicate counters so their simulation in matching is fast.} 

\paragraph{Related work.}
In the context of regex matching, counting automata were used in several forms under several names (e.g. \cite{aplas19,oopsla,cikm15,GGM12,SEJ08,jha_extended,Sperberg-McQueen-ExtendedFAWeb,Gelade09_countingregex,Hovland-Membership2012}).
Besides \cite{oopsla} discussed above, other solutions to matching of counting regexes \cite{GGM12,cikm15,Hovland09,KilTuh07,SEJ08,aplas19,KT03} handle small classes of regexes 
or do not allow matching linear in the text size and independent of counter bounds. 
%
%
The work \cite{aplas19} proposes a CA-to-CA determinization
producing smaller automata than the explicit CA determinization 
for the limited class of monadic regexes, covered by letter-marked counting,
and the size of their deterministic automata is still dependent on the counter bounds. 
%
%
The work \cite{cikm15} uses a notion of automata with counters of
\cite{GGM12}.
It focuses mostly on deterministic regexes, a class much smaller than regexes with synchronizing counting, and proposes a matching algorithm still dependent on the counter bounds.
The paper \cite{KT03} proposes an  algorithm that takes time at worst quadratic to the length of the text. 
Extended FA (XFA) of \cite{SEJ08,jha_extended} augment NFA with a scratch memory of
bits that can represent counters, and their determinization is exponential in counter bounds already for regexes such as
\regex{.*a.\{$k$\}}.
The \emph{counter-1-unambiguous}
regexes of \cite{Hovland09,Hovland-Membership2012} can be directly compiled into
deterministic automata called FACs, similar to our CA, independent of counter bounds, but the class
is limited, excluding e.g., \regex{.*a.\{$k$\}}.  

\section{Preliminaries}
\label{sec:prelims}
We use $\nat$ to denote the natural numbers including 0. 
For a set $S$, $\PowerSet{S}$ denotes its powerset and $\PowerSetFin{S}$ is the set of all \emph{finite} subsets of $S$.

A \emph{first order language (f.o.l.)}
$\fol = (\fset,\pset)$ consists of a set of \emph{function symbols} $\fset$ and a set of \emph{predicate symbols} $\pset$. 
An \emph{interpretation} $\interp$ of $\fol$ with a \emph{domain} $D_\interp$ assigns a function $\interpof f:D_\interp^n\rightarrow D_\interp$ to each $n$-ary $f\in \fset$ and a function $\interpof p:D_\interp^n\rightarrow \tf$ to each $n$-ary $p\in\pset$.
An \emph{assignment} of a set of variables $\vset$ in $\interp$ is a total function $\vass:\vset\rightarrow D_\interp$.  
The set of \emph{terms} $\termsof \fol \vset$ 
and the set $\formulasof \fol \vset$ of \emph{quantifier free formulae} (boolean combinations of atomic formulae) over $\fol$ and $X$, 
as well as the interpretation of a term, $\interpof t(\vass)$, and a formula, $\interpof \varphi(\vass)$, 
are defined as usual.
%
We denote by $\vass \models_\interp \varphi$ that the formula $\varphi$ is \emph{satisfied} (interpreted as true) by the assignment $\vass$. It is then \emph{satisfiable}.
We drop the sub/superscript $\interp$ when it is clear from the context. 
We write $\varphi[x]$ and $t[x]$ to denote a unary formula $\varphi$ or term $t$, respectively, with the free variable $x$, and we may also abuse this notation to denote the term/formula with its only free variable replaced by $x$.
We write $\interpof t(k)$ and $\interpof\varphi(k)$ to denote the values $\interpof t(\{x\mapsto k\})$ and $\interpof \varphi(\{x\mapsto k\})$.
For a set of formulae 
$\Psi = \set{\psi_1,\ldots,\psi_n}$, 
the set $\Minterms{\Psi}$ consists of all \emph{minterms} of $\Psi$, satisfiable conjunctions
$\phi_1 \land\cdots\land \phi_n$ 
\mbox{where for each $i:1\leq i \leq n$, $\phi_i$ is $\psi_i$ or $\neg\psi_i$.}

We fix a finite \emph{alphabet} $\Sigma$ of \emph{symbols/letters} for the rest of the paper. 
Words are sequences of letters,
with the \emph{empty word} $\epsilon$. The \emph{concatenation} of words $u$ and $v$ is denoted $u\concat v$,  $uv$ for short. A set of words over $\Sigma$ is a \emph{language}, the concatenation of languages is $L\concat L' = \{u\concat v\mid u\in L \land v \in L'\}$, $LL'$ for short.
\emph{Bounded iteration} $x^i$, $i\in \nat$, of a word or a language $x$ is defined by $x^0 = \epsilon$ for a word, $x^0 = \set{\epsilon}$ for a language, and $x^{i+1} = x^i\concat x$. 
Then $x^* = \bigcup_{i\in\nat}x^i$. 
We consider a usual basic syntax of \emph{regular expressions (regexes)}, generated by the grammar 
$
R \ ::= \
\rexeps \ \mid \ 
\rexa \ \mid \ 
\rexbracket R \ \mid \ 
\rexconcat R R \ \mid \ 
\rexuni R R \ \mid \ 
\rexstar R \ \mid \
\rexcount R {m,n}  
$
where $m\in \nat$, $n \in \nat \cup \infty$, $0\leq m$, $0<n$, $m\leq n$,
and $\rexa\in\Sigma$.
We use $\rexcount{R}{m}$ for $\rexcount{R}{m,m}$.
Regexes containing a sub-expression with the \emph{counter} $\rexcount R {m,n}$ or $\rexcount R {m}$ are called \emph{counting regexes} and $m,n$ are \emph{counter bounds}.
We denote by $\boundof R$ the maximum integer occurring in the counter bounds of regex $R$ 
and we denote the number of counters by $\cntof R$.
A~regex with \emph{flat counting} does not have nested counting, that is, in a sub-regex $S\{m,n\}$, $S$ cannot contain counting.  
%
%
The \emph{language} of a regex $R$ is constructed inductively to the structure: $L(\rexeps) = \{\epsilon\}$, $L(\rexa) = \{a\}$ for $a\in \Sigma$, $L(\rexconcat R {R'}) = L(R)\concat L(R')$, $L(\rexstar R) = L(R)^*$, $L(\rexuni R {R'}) = L(R)\cup L(R')$, and
$L(\rexcount R {m,n}) = \bigcup_{m\leq i\leq n} L(R)^i$. 
We understand $|R|$ simply as the length of the defining string, e.g. $|\regex{(ab)\{10\}}| = 8$.
We define $\sharp R$ as the number of character occurrences in $R$, formally, $\sharp a = 1$ for $a\in\Sigma$, $\sharp \epsilon = 0$, $\sharp\!\regex{(}R\regex{)} = \sharp R\regex{\{m,n\}} = \sharp R$, and $\sharp R\cdot S = \sharp R \regex{|} S = \sharp R + \sharp S$.

A \emph{(nondeterministic) automaton (NA)} is a tuple $A = (Q,\Delta,I,F)$ where $Q$ is a set of \emph{states},
$\Delta$ is a set of \emph{transitions} of the form $\move q a r$ with $q,r\in
Q$ and $a\in\Sigma$, $I\subseteq Q$ is the set of \emph{initial states}, and $F\subseteq Q$
is the set of \emph{final states}. 
A run of $A$ over a word $w = a_1 \ldots a_n$ from state $p_0$ to $p_n$, $n \geq 0$
is a sequence of transitions
 $\move {p_0} {a_1} {p_1}$,
 $\move {p_1} {a_2} {p_2}$, $\ldots$,
 $\move {p_{n-1}} {a_n} {p_n}$ from $\Delta$. The empty sequence is a run with $p_0 = p_n$ over $\epsilon$.
The run is \emph{accepting} if $p_0 \in I$ and $p_n\in F$, and the language $L(A)$ of $A$ is the set of all words for which $A$ has an accepting run.
A state $q$ is \emph{reachable} if there is a run from $I$ to it.
The \emph{size} of the NA, $|A|$, is defined as the number of its states plus the number of its transitions.
The automaton is \emph{deterministic (DA)} iff $|I| = 1$ and for every state $q$ and symbol $a$, 
$\Delta$ has at most one transition $\move q a r$.
The \emph{subset construction} transforms the NA to the DA with the same language
$\dfaof A = (\sc Q,\sc{\Delta\!}, \sc I, \sc F)$ 
where $\sc Q \subseteq \pow Q$ and $\sc{\Delta\!}$ are the smallest sets of states and transitions satisfying
$\sc I = \{I\}$, 
$\sc {\Delta\!}$ has for each $a\in\Sigma$ and each $S\in \sc Q$ the transition  
$\move S a {\{s' \mid s\in S \land \move s a {s'}\in\Delta\}}$,
and $\sc F = \{S \in \sc Q\mid S\cap F\neq\emptyset\}$.
%
When the set of states $Q$ is finite, we talk about (deterministic) \mbox{\emph{finite state} automata (NFA, DFA).}\footnote{We do not require finiteness in the basic definition in order to avoid artificial restrictions of the notions of automata with registers/counters/counting sets defined later.}





This paper is concerned with the problem of fast \emph{pattern matching}, basically a membership test: given a regex $R$ and a text $w$, decide whether $w\in L(R)$.
While $w$ may be very long, $R$ is normally small, hence the dependence on $|w|$ is the major factor in the complexity.   
The offline DFA simulation takes time linear in $|w|$. It (1) compiles $R$ into an NFA $\nfaof R$ (2) determinizes it, and (3) follows the DFA run over $w$ (aka \emph{simulates} the DFA on $w$), all in 
time and space $\Theta(2^{|\nfaof R|}+|w|)$. The cost of determinization, exponential in $|\nfaof R|$, is however too impractical.  
Modern matchers such as Grep or RE2 \cite{grep,re2} therefore use the techniques of online DFA simulation, 
where only the part of the DFA used for processing $w$ is constructed. It reduces the complexity to $O(\min(2^{|\nfaof R|}+|w|, |w|\cdot|\nfaof R|))$  
(the first operand of $\min$ is the explicit determinization in case the entire DFA is constructed, plus the cost of DFA-simulation; the second operand is the cost of the online-DFA simulation, coming from that every step may incur construction of a new DFA state and transition in time $O(|\nfaof R|)$).
For counting regexes, the factor $|\nfaof R|$ depends linearly (or more if counting is nested) on $\boundof R$ and thus exponentially on $|R|$. 
This makes counting very problematic in practice \cite{oopsla,usenix,Sperberg-McQueen-ExtendedFAWeb}.
We will present a matching algorithm which is \emph{fast} for a specific class of regexes, meaning that its run-time is still linear in $|w|$ but is independent of $\boundof R$.

\section{Counting Automata}
\label{sec:ca}



We use a rephrased definition of counting automata and counting-set automata of \cite{oopsla}.
We will present them as a special case of a generic notion of automata with registers.
\begin{definition}[Automata with registers]
An \emph{automaton with registers} (RA) operated through an f.o.l. $\fol$ under an interpretation $\interp$ 
is a tuple 
 $A = (\vset, Q, \Delta, I, F)$ where 
 $\vset$ is a set of variables called \emph{registers};
 $Q$ is a finite set of \emph{states}; 
 $\Delta$ is a finite set of \emph{transitions} of the form $\move{q}{a, \varphi, u}{p}$ where
 $p,q\in Q$, $a\in\Sigma$, 
 $u:\vset\rightarrow \termsof \fol \vset$ is an \emph{update},
 and $\varphi \in \formulasof \fol \vset$ is a \emph{guard};  
 $I$ is a set of \emph{initial configurations}, where a \emph{configuration} is a pair of the form $(q, \mem)$ where $q \in Q$ and $\mem:\vset \rightarrow D_\interp$ is a register assignment called a \emph{memory};
 and 
 $F:Q \to \formulasof \fol\vset$ is a \emph{final condition assignment}. 
 

The language of $A$, $L(A)$, is defined as the language of its \emph{configuration automaton} $\confaut A$.  
States of $\confaut A$ are \emph{configurations} of $A$ that are reachable. $I$ is the set of initial states of $\confaut A$. 
It has a transition $\move{(q, \mem)}{a}{(q',\mem')}$ iff $(q, \mem)$ is reachable and $A$ has a transition $\delta = \move{q}{a,\varphi,u}{q'} \in \Delta$ such that $(q',\mem')$ is the \emph{image} of $(q, \mem)$ under $\delta$, denoted $(q', \mem') = \delta(q,\mem)$, meaning that (1) $\delta$ is \emph{enabled} in $(q, \mem)$, $\mem \models \varphi$, 
and (2) $\mem' = u(\mem)$, i.e. $\mem'(x) = u(x)^{\interp}(\mem)$ for each $x \in X$. 
We let $\delta(C) = \{\delta(c)\mid c\in C\}$ for a set of configurations $C$. 
A configuration $(q, \mem)$ is a final if $\mem \models F(q)$.
By \emph{runs of} $A$ we mean runs of $\confaut A$.
 The RA $A$ is \emph{deterministic} if $\confaut A$ is deterministic.
 The size of the RA is $|A| = |Q| + \sum_{\delta\in\Delta}|\delta|$ where $|\delta|$ is the sum of the sizes of the update and the guard.
 \end{definition}

\begin{definition}[Counting automata]
A \emph{counting automaton} (CA) is an automaton with registers, called \emph{counters}, operated through the \emph{counting language} $\clang$ 
that contains the unary increment function, denoted $x+1$, constants $0$ and $1$, and predicates $x > k$ and $x\leq k$, $k \in \nat$,
with the standard interpretation over natural \mbox{numbers, that we denote $\cinterp$.}
\end{definition}

\begin{wrapfigure}[10]{r}{0.5\textwidth}
\vspace{-6mm}
  \centering%
  \begin{tikzpicture}[shorten >=1pt,node distance=1.5cm,on grid,initial text={\scriptsize$x \ass 0$},state/.style={circle, draw, minimum size=0.3cm}]
    \node[state,initial] (q_0)                  {\scriptsize$q_0$};
    \node                (help) [right=of q_0]  {};
    \node[state]         (a_1)  [above=0.6cm of help] {\scriptsize$a_1$};
    \node[state]         (b_1)  [below=0.6cm of help] {\scriptsize$b_1$};
    \node[state,accepting,label={right:\scriptsize$[x \geq 3]$}]         (b_2)  [right=of help] {\scriptsize$b_2$};
    \path[->] (q_0) edge node [above,rotate=22] {\scriptsize$a$} node [below,rotate=22] {\scriptsize$x \ass 1$} (a_1)
              (q_0) edge node [above,rotate=-22] {\scriptsize$b$} node [below,rotate=-22] {\scriptsize$x \ass 1$} (b_1)
              (a_1) edge node [above,rotate=-22] {\scriptsize$b$} node [below,rotate=-22] {\scriptsize$x \ass x$} (b_2)
              (b_1) edge node [above,rotate=22] {\scriptsize$b$} node [below,rotate=22] {\scriptsize$x \ass x$} (b_2)
              (b_2) edge [bend right=50] node [right,pos=0.5] {\vspace{0.2cm}\scriptsize$a; x < 8$} node [right,pos=0.2] {\scriptsize$x \ass x+1$} (a_1)
              (b_2) edge [bend left=50] node [right,pos=0.2] {\scriptsize$b; x < 8$} node [right,pos=0.5] {\vspace{0.2cm}\scriptsize$x \ass x+1$} (b_1);
  \end{tikzpicture}%
  \caption{$\caof{R}$ for $R = \regex{((a|b)b)\{3,8\}}$. The accepting condition of all states is $\bot$ except for $b_2$ whose accepting condition is written in the square brackets.}%
  \label{fig:exCAR}%
\end{wrapfigure}
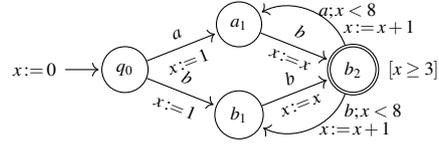
Regexes with counting may be translated to CA by several methods (\cite{oopsla,Sperberg-McQueen-ExtendedFAWeb,Gelade09_countingregex,Hovland-Membership2012}). 
We use a slightly adapted version of \cite{Gelade09_countingregex}---an extension of Glushkov's algorithm \cite{Glu61} to counting.
%
For a regex $R$, it produces a CA $\caof R = (X,Q,\Delta,\{\alpha_0\},F)$. 
\cref{fig:exCAR} shows an example of such CA.
The construction is 
discussed in detail in \cref{sec:caproperties}, here we only overview the important properties needed in Sections~\ref{sec:determinization}-\ref{sec:augmented}: 
\begin{enumerate} 
\item
Every occurrence $S$ of a counted sub-expression $\rexcount T {\mnof S,\mxof S}$ of $R$ corresponds to a unique counter $x_S$ and a substructure $A_S$ of $\caof R$. Outside $A_S$, $x_S$ is inactive (a dead variable) and its value is 0, it is assigned 1 on entering $A_S$, and every iteration through $A_S$ increments the value of $x_S$ while reading a word from $L(T)$. Our minor modification of~\cite{Gelade09_countingregex} is related to the fact that the original assigns $1$ to inactive counters while we need $0$.
\item
$\caof R$ has at most $\sharp R+1$ states, $\cntof R.\sharp R^2$ transitions, $\cntof R$ counters. It has at most $\sharp R^2$ transitions if $R$ is flat. 
\item
$\caof R$ has a single initial configuration $\alpha_0 = (q_0, \smem_0)$ s.t. $\smem_0(x_S) = 0$ for each $x_S \in X$.
\item
Guards and final conditions are conjunctions consisting of at most one conjunct  of the form $\mnof S\leq x_S$ or $\mxof S > x_S$ per counter $x_S \in X$. 
A transition update may assign to $x_S\in X$ only one of the terms $0$, $1$, $x_S$, and $x_S+1$. 
It has no guard on $x_S$ if it is assigned $x_S$, i.e. kept unchanged,
it has the guard $x_S\geq \mnof S$ iff $x_S$ is reset to $0$ or $1$
(a counter cannot be reset before reaching its lower bound),
and it has the guard $x_S < \mxof S$ iff $x_S$ is assigned $x_S+1$
(counter can never exceed its maximum value $\mxof S$).   
Hence, a counter can never exceed $\boundof R$. 
\item
Flatness of $R$ translates to the fact that configurations of $\caof R$ assign a non-zero value to at most one counter. 
This implies that 
$\confaut{\caof R}$ has at most $|Q|.\boundof R$ states and also that
$\caof R$ is \emph{Cartesian},
a property that will be defined in \cref{sec:determinization} and is crucial for correctness of our CA determinization (\cref{theo:complexity} in \cref{sec:augmented}.) 
\end{enumerate}
A~DFA can be obtained by the subset construction in the form $\dfaof{\confaut {\caof R}}$, called \emph{explicit determinization}.
Due to the factor $\boundof R$ in the size of $\confaut{\caof R}$, the explicit determinization is exponential to $\boundof R$ even if $R$ is flat, meaning doubly exponential to $|R|$ ($R$ has $\boundof R$ written as a decadic numeral). 
If $R$ is not flat, 
then 
the factor $\boundof R$ is replaced by $\expof R$.

\section{Counter-subset Construction}
\label{sec:determinization}

In this section, we formulate a modified version 
of determinization of CA from \cite{oopsla} that constructs a machine of a size independent of $\boundof R$.
Our version handles the entire class of Cartesian CA (defined below) and in turn also all regexes with flat counting.

The main idea of the determinization remains the same as in \cite{oopsla}. The standard subset construction is augmented with registers, we call them \emph{counting sets}, that can store sets of counter values that would be generated by non-deterministic runs of the CA. 
The automata with counting-sets as registers are called \emph{counting-set automata}.
Our first modification of \cite{oopsla} is indexing of counters by states. 
%
In intuitively, 
this allows to handle cases such as $\regex{a*(ba|ab)\{5\}}$, where, after reading the first $ab$, the counter is either incremented or not ($b$ is the first letter of the counted sub-expression or not). This would violate the uniformity property of CA necessary in \cite{oopsla}---the set of values generated by the non-deterministic CA runs must be the same for every CA state. 
In our modified version, values at distinct states are stored separately in registers indexed by those states and may differ. 
%
%
%
Then, in order to handle the indexed counters, we have to introduce 
a general assignment of counters, 
allowing to assign the \emph{union} of other counters.\footnote{\cite{oopsla} could assign to a counter $x$ only a constant or function of the current value of $x$.} 
Intuitively, when a run non-deterministically branches into two states, each branch needs to continue with its \emph{own copy} of the set stored in a counter indexed by the state of the branch.
The union of sets is used when the branches join again. 
This brings a technical challenge that we solve in this work: how to simulate the counting-set automata fast when the set union and copy are used? The solution is presented in Sections~\ref{sec:matching} and \ref{sec:augmented}.

\begin{definition}[Counting-set automata]
A \emph{counting-set automaton (CSA)} is an automaton with registers operated through the \emph{counting-set language} $\slang$ under the \emph{num\-ber-set interpretation} $\sc\cinterp$ 
where the language $\slang$
extends the counting language $\clang$  
with the constant $\emptyset$, binary union $\cup$, and set-filter functions $\filterof{p}$ where
$p$ is a predicate symbol of $\clang$. 
For simplicity, we restrict terms assigned to counters by transition updates to the form  $t = t_1 \cup \cdots \cup t_n$ where each $t_i$  is either (a)~a~term of $\clang$ or $\emptyset$, (b) of the form $\filter_{{p}(t')}$ where $t'$ is a term of $\clang$.
Each $t_i$ is called an \emph{$r$-term} of $t$.


The domain of $\jinterp$ is \emph{sets of natural numbers}, $\pow{\nat}$. 
The interpretation of the predicates and functions of $\clang$ under $\jinterp$ is derived from the base number interpretation of the same predicates and functions:
A function returns the image of the set in the argument under the base semantics,
$\jinterpof f(S) = \set{\cinterpof f(n) \mid n \in S}$.
A set satisfies a predicate if some of its elements satisfy the base semantics of that predicate,  
$\jinterpof p(S) \iff \exists e \in S : \cinterpof p(e)$.
Filters then filter out values that do not satisfy the base semantics of their predicate,   
$\jinterpof\filter_{\!\!p}(S) = \set{e \in  S\mid \cinterpof p(e)}$.
Finally, $\emptyset$ is interpreted as the empty set and $\cup$ as the union of sets.
We denote memories of the CSA by $\smem$ to distinguish them from memories of CA.
We write DCSA to abbreviate deterministic CSA. 
\end{definition}
%

Less formally, registers of CSA hold sets of numbers and are manipulated by the increment $x+1$ of all values, 
assignment of constant sets $\{0\}$, $\{1\}$, and $\emptyset$, denoted by $0$, $1$, and $\emptyset$, filtering out values smaller or larger than a constant, denoted $\filter_{x\leq k}(x)$ and $\filter_{x < k}(x)$, and testing on a presence of a value $x$ satisfying $x\leq k$ or $x < k$, $k\in\nat$.


We will present an algorithm that determinizes a CA $A = (\vset,Q,\Delta,I,F)$, fixed for the rest of the section, into a DCSA $\scaut A = (\sc \vset, \sc Q,\sc{\Delta\!},\sc I,\sc F)$.
We assume that guards of transitions in $\Delta$ and final conditions are of the form $\bigwedge_{x\in Y} p_x[x],Y\subseteq X$, i.e. conjunctions with a at most a single atomic predicate per counter.
This is satisfied by all $\caof R$, for any regex $R$ (see the list of properties of $\caof R$ in \cref{sec:ca}).\footnote{Every CA can be transformed to this form by transforming the formulae to DNF and creating clones of transitions/states for individual clauses.} 

Runs of $\scaut A$ will \emph{encode} runs of $\dfaof {\confaut A}$ obtained from the explicit determinization of $A$.
Recall that the states $\dfaof {\confaut A}$ are sets of configurations of $A$, pairs $(q,\mem)$ of a state and a counter assignment. $\scaut A$ will represent the sets of counter values within a DA state as run-time values of its registers. 

Particularly, for every state $q$ and a counter $x$ of the CA,
$\scaut A$ has a register $x_q$ in which it remembers, after reading a word $w$, the set of all values that $x$ reaches in runs of the base CA on $w$ ending in $q$. Hence, we have $\sc \vset = \{x_q \mid x\in \vset \land q\in Q\}$
\begin{definition}[Encoding of sets of CA configurations]
A state $S = \{(q_i,\mem_i)\}_{i=1}^n$ of $\dfaof {\confaut A}$ is encoded as the $\scaut{A}$ configuration 
$\encode(S) = (\{q_i\}_{i=1}^n,\smem)$ where $\smem(x_q) = \{\mem_i(x) \mid q_i = q\}_{i=1}^n$.
\end{definition}

Since a set of assignments appearing with the state $q$ is broken down to sets of values of the individual counters, it disregards relations between values of different counters. 
For instance, in the DA state $S_1 = \{(q, \{x \mapsto 0, y \mapsto 0\}), (q, \{x \mapsto 1, y \mapsto 1\})\}$, the values of $x$ and $y$ are either both 0 or both 1, but 
$\encode(S_1) = (q, \{x_q \mapsto \{0,1\}, y_q \mapsto \{0,1\}\})$ does not retain this information.
It is identical to the encoding of another DA state $S_2 = \{(q, \{x \mapsto 1, y \mapsto 0\}), (q, \{x \mapsto 0, y \mapsto 1\})\}$.
This is the same loss of information as in the so-called Cartesian abstraction.
The encoding is hence precise and unambiguous only when we assume that inside the states of  $\dfaof A$, the relations between counters are always unrestricted---there is no information to be lost.
We then call the CA \emph{Cartesian}, as defined below. 
The encoding function is then unambiguous, and we call the inverse function \emph{decoding}, denoted $\decode$.  

\begin{definition}[Cartesian CA]
Assuming the set of counters of $A$ is $X = \{x_i\}_{i=1}^m$,
then a set $C$ of configurations of $A$ is \emph{Cartesian}
iff, for every state $q$ of $A$, 
there exist sets $N_1,\ldots,N_m\subseteq \nat$ such that 
%
%
$(q,\{x_i\mapsto n_i\}_{i=1}^m) \in C$
iff
$(n_1,\ldots, n_m) \in N_1\times\cdots\times N_m$.
The CA $A$ is \emph{Cartesian} iff all states of $\dfaof{\confaut A}$ are Cartesian.
\end{definition}
For instance, the DA states $S_1$ and $S_2$ above are not Cartesian, while $S_1\cup S_2$ is.

Similarly as the regex to CA construction of \cite{oopsla}, 
our regex to CA construction discussed in \cref{sec:ca} (App. \ref{sec:caproperties}) returns a Cartesian CA when called on a flat regex.

%

\paragraph{Subset construction for Cartesian CA.}
The algorithm below is a generalization of the subset construction. 
%
Let us denote by $\indexby t q$ the term that arises from $t$ by replacing every variable $x\in \vset$ by $x_q$, analogously $\indexby \varphi q$ for formulas. 
%
We have $\sc Q \subseteq \pow Q$,
the initial configuration $\sc I = \{\encode(I)\}$,
and the final conditions assign to $R\in \sc Q$ the disjunction of the final conditions of its elements, $\sc F(R) = \bigvee_{q \in R} \indexby{F(q)} q$.

We will construct $\scaut A$ which is deterministic and its runs encode the runs of DA $\dfaof{\confaut A}$. $\confaut{\scaut A}$ will be isomorphic to $\dfaof{\confaut A}$.  
For that, we need for each transition $\delta$ of $\dfaof{\confaut A}$ one unique transition of $\scaut A$ over the same letter enabled in the encoding 
of the source of $\delta$ and generating the encoding 
of the target of $\delta$. 
In other words, we need 
 for each transition $\move {\decode(R,\smem)} a {\decode(R',\smem')}$ of $\dfaof{\confaut A}$ one unique transition $\delta' = \move  {R} {a,\varphi,u} {R'} \in \sc{\Delta\!}$ with $(R',\smem') = \delta'(R,\smem)$. 
That transition $\delta'$ will be built by summarizing the effect of all base CA $a$-transitions enabled in the CA configurations of $\decode(R,\smem)$. 

To construct the transition $\delta'$, we first translate each base transition $\delta = \move q {a,\varphi_\delta,u_\delta} r\in\Delta$ into its set-version $\sc \delta$,
%
supposed to transform an encoding of a (Cartesian) set $C$ of configurations, $\encode(C)$, into the encoding of the set of their images under $\delta$, $\encode(\delta(C))$, and enabled if $\delta$ is enabled for at least one configuration in $C$.
To that end, assuming $\varphi_\delta = \bigwedge_{x\in \vset}p_x[x]$, we
(1) construct the update $u_\delta^{\!\filter}$ from $u_\delta$ by substituting in every $u_\delta(x), x\in \vset$ variables $y\in \vset$ by their filtered versions $\filterof {p_y}(y)$,
(2) add indices to registers that mark the current state, resulting in the transition
$\sc \delta = \move q {a,\sc \varphi_\delta,\sc u_\delta}  {r}$ where $\sc \varphi_\delta = \indexby{\varphi_\delta} q$ and 
$\sc u_\delta$ assigns to every $x_r,x\in X$ the term $ \indexby {u_\delta^{\!\filter}(x)} q$.
%
%

The states $\sc Q$ and the transitions $\sc {\Delta\!}$ are then constructed as the smallest sets satisfying that $\encode(I)\in \sc Q$ 
and
every $R\in \sc Q$ has for every $a\in\Sigma$ the outgoing transitions constructed as follows.
Let $\{\move {q_j} {a,\varphi_j,u_j} {r_j}\}_{j\in J}$ for some index set $J$ be the set of \emph{constituent $a$-transitions} for $R$, all $a$-transitions $\sc \delta$ where $\delta\in\Delta$ originates in $R$.
To achieve determinism, $\sc {\Delta\!}$ has the transition $\move R {a,\psi,u} {R'}$ for every minterm
$\psi \in \Minterms{\{\varphi_j\}_{j\in J}}$. 
The update $u$ and target $R'$ are constructed 
from the set $\{\move {q_j} {a,\varphi_j,u_j} {r_j}\}_{j\in K}$, $K\subseteq J$, of constituent transitions with guards $\varphi_j$ compatible with the minterm $\psi$, i.e., with satisfiable $\psi\land\varphi_j$. $R'$ is the set of their target states, $R' = \{r_j\}_{j\in K}$,
and
$u(x)$ unites all their update terms $u_j(x)$,
i.e. $u(x) = \bigcup_{j\in K} u_j(x)$, for each $x\in \sc\vset$.

\begin{example}
\label{example:determ}
When showing examples of transition updates, we write $x\ass t$ to denote that $u(x)=t$ and we omit the assignments $x\ass\emptyset$ in CSA. 

Let $R = \{p,q\}$ and let the $a$-transitions originating at $R$ be 
$\move q {a,\top,x\ass x} s$, \\
$\move p {a,x<n,x\ass x+1} r$, and    
$\move p {a,x\geq m,x \ass 1} s$.
They induce three constituent transitions for $R$ and $a$, 
$\move q {a,\top,x_{s} \ass x_q} s$,
$\move p {a,x_p<n,x_{r} \ass \filterof {x<n}(x_p)+1} r$, and
$\move p {a,x_p\geq m,x_{s} \ass 1} s$.             
A transition $\move{R}{a,\psi,u'}{R'}$ is constructed for each of the following minterms $\psi$:
$x_p{<}n \land x_p{\geq} m$,
$\neg x_p{<}n \land x_p{\geq} m$,
$x_p{<}n \land \neg x_p{\geq} m$,
$\neg x_p{<}n \land \neg x_p{\geq} m$.
For the first one, all three constituent transitions are compatible and so the update $u'$ is 
$x_r \ass \filterof {x<n}(x_p)+1; x_s \ass x_q \cup 1$ (update of $x_r$ is taken from the first constituent transitions leading to $r$, update of $x_s$ is the union of the updates of the second two transitions leading to $s$) and the target state is $R' = \{r,s\}$.
%
\qed
\end{example}

$\scaut{A}$ is deterministic since it has a single initial configuration and the guards of transitions originating in the same state are minterms.  
The size of $\scaut A$ obviously depends only on the size of $A$ and not on the interpretation of the language. Especially, when $A$ is $\caof R$ for some regex $R$, the size does not depend on $\boundof R$.  
The theorem below is proved in \cref{app:alg1correctnessproof}.%
\footnote{
It may be interesting to note that, as follows from our formulation of the determinization, the construction is independent of the particular f.o.l. used to manipulate registers and of its interpretation.  
The determinization could be applied to any kind of automata that fits the definition of automata with registers. 
The numbers could be manipulated by other functions and tests, natural numbers could be replaced by reals etc.
The counting-set automata are themselves an instance of automata with registers.
One could also think about push-down automata or, with small modifications, variants of data-word automata with registers. 
}

\begin{restatable}{theorem}{correctness}
\label{theo:alg1corectness}
$\scaut A$ is deterministic,  $|\scaut A|\in O(2^{|A|})$, and if $A$ is Cartesian, then $L(A) = L(\scaut A)$. 
\end{restatable}
Since for regexes with flat counting, our regex to CA algorithm always returns a Cartesian CA, we can transform them into DCSA.

\section{Fast Simulation of Counting-set Automata}

\label{sec:matching}


In this section, we discuss how a run of a DCSA on a given word can be \emph{simulated} efficiently to achieve fast matching.
Let us fix a word $w = a_1\cdots a_n$ together with the DCSA 
$A = (X, Q, \Delta, \{\alpha_0\}, F)$. 
We wish to construct the run of the DCSA on $w$ and test whether the reached configuration is accepting.
We aim at a running time linear to $|w|$ and independent of the sizes of the sets stored in $A$'s registers at run-time.  

We will assume that the initial configuration $\alpha_0$ of $A$ assigns to every register a singleton or the empty set. The assumption is satisfied by CSA constructed from $\caof R$, $R$ being any regex, by the algorithms of \cref{sec:determinization} and also \cref{sec:augmented}.\footnote{This is a technical assumption important in order for unions of the initial sets not to influence the overall complexity of the simulation.}

Technically, the simulation maintains a configuration $\conf = (q,\smem)$, initialized with $\alpha_0$,
and for every $i$ from $1$ to $n$, 
it constructs the transition $\move {\conf} {a_i} {\conf'}$ of $\confaut{A}$ and replaces $\conf$ by the successor configuration $\conf' = (q',\smem')$. 
We use the key ingredient of fast simulation from \cite{oopsla}, the \emph{offset-list data structure} for sets of numbers with constant time addition of 0/1, comparison of the maximum to a constant, reset, and increment of all values. 
The problem is that the newly added union and copy of sets are still linear to the size of the sets, and hence linear to the maximum counter bounds.
%
We show how, under a condition introduced below, set copy can be avoided entirely and the cost of union can be amortized by the cost of incrementing the sets.
This will again allow a CSA-simulation in time independent of $\boundof A$ and falling into $\bigo{|A|\cdot|w|}$.

First, we define a property of CSA sufficient for fast simulation---that the updates on its transitions do not \emph{replicate counters}.

 
\begin{definition}[Counter replication]
We say that a CSA \emph{replicates counters} if for some transition $\move{q}{a,\phi,u}{r}$, some counter appears in the image of $u$ twice,
that is, it appears in two r-terms of some $u(x)$ or it appears in $u(x)$ as well as in $u(y)$ for $x\neq y$. 
A \emph{non-replicating} CSA does not replicate counters. 
\end{definition}
For instance, $\{x\mapsto x;y\mapsto x+1\}$ and $\{x\mapsto x\cup x+1,y\mapsto y\}$ are updates where $x$ is replicated, $\{x\mapsto x+1,y\mapsto y\}$ is not a replicating update.


\paragraph{Offset-list data structure.}
The \emph{offset-list} data structure of \cite{oopsla} allows constant time implementation of the set operations of increment of all elements, reset to $\emptyset$ or $\{0\}$ or $\{1\}$, addition of $0$ or $1$, and comparison of the maximum with a constant. 

It assigns to every counter $x\in X$ a pointer $\olof x$ to an \emph{offset-list pair} $(o_x,l_x)$ with the \emph{offset} $o_x\in \nat$ and a sorted list $l_x = m_1,\ldots,m_k$ of integers. The data structure implementing the list needs constant access to the first and the last element, forward and backward iteration of a pointer, and insertion/deletion at/before a pointer to an element. This is satisfied for instance by a doubly-linked list that maintains pointers to the first and the last element. The offset-list pair represents the set $\smem(x) = \{m_1+o_x,\ldots, m_k+o_x\}$. Union of two such sets is still linear in their size, but we will show that if the CSA does not replicate counters, the cost of set unions can be amortized by the cost of increments. 

\paragraph{Finding the CSA transition and evaluating the update.}
The first step of computing $\conf'$ from $\conf$ is finding the transition $\move{q}{a_i,\phi,u}{q'}\in \Delta$, the only $a_i$-transition from $q$ that is enabled, i.e. where $\smem\models\phi$. 
The simplest algorithm iterates through the transitions of $\Delta$ and, for each of them, tests whether $\smem$ satisfies its guard. 
The cost of evaluating an atomic counter predicate $p$, i.e., deciding whether $\smem\models p$, is constant: since the lists $l_x$ are sorted, we only need to access the first or the last element and the offset to decide $x < n$ or $x \geq n$, respectively.
With that, the cost of evaluating $\phi$ is linear to the size of $\phi$. 
The cost of the iteration through the transitions of $\Delta$ is then  linear in the sum of their sizes, which is within $\bigo{|A|}$.\footnote{In the later sections, we will show better complexity bounds of finding the transition under additional assumptions on the CSA guards that are produced by our constructions.}

Having found $\move{q}{a_i,\phi,u}{q'}$, we evaluate its update to compute $\smem'$ and compute $\conf'$ as $(q',\smem')$. We will explain the algorithm and argue that the amortized cost of computing $\smem'$ is in $\bigo {|X|}$. 
The update is evaluated by, for each $x\in X$, evaluating all r-terms in $u(x)$, uniting the results, and assigning the union to $\olof x$.

First, we argue that evaluating an r-term $t$ of $u(x)$, i.e. computing $t(\smem)$, is amortized constant time. 
Since the counters are non-replicating, we can compute the value of each r-term $t[y]$ in situ. That is, we modify the offset-list pair $(o_y,l_y)$ and return the pointer $\olof y$. 
%
%
The original value of $y$ can be discarded after evaluating $t[y]$ since $y$ does not appear in any other r-term.  
There are 5 cases:
(1)
If $t$ is $0$ or $1$, then we return a pointer to a fresh offset-list pair with the offset $0$ and the list containing only $0$ or $1$, respectively. 
This is done in constant time. 

(2)
If $t$ is $y\in Y$, then we return $\olof y$. 

(3)
If $t$ is $y+1$, then $o_y$ is incremented by one.  
This constant time implementation of the increment is the reason for pairing the lists with the offsets. 

(4)
If $t$ is $\filterof {p} [y]$, then $l_y$ is filtered by the atomic predicate $p$. 
Filtering with the predicate $x \geq n$ uses the invariant of sortedness of $l_y$.
It is done by iterating the following steps: i)~test whether the list head is smaller than $n - o_y$ and ii) if yes, remove the head, if not, terminate the iteration. 
Every iteration is constant time: The cost of the iterations which remove an element is amortized by the cost of additions of the element to the list. 
What remains is only the constant cost of the last iteration which detects an element greater or equal to $n - o_y$, or that the list is empty. 
Filtering with $x < n$ is analogous (the iterations test and remove the last element instead of the head).

(5)
If $t$ is $\filterof p (y) + 1$, then the construction for the constant increment is applied after the constant filter discussed above.

Next, we argue that computing the union of values of the r-terms in $u(x)$ may be amortized by the cost of evaluating the increment terms. 
Let $\ol_1,\ldots,\ol_n$ be the offset-list representations of the values of the terms in $u(x)$ computed by the algorithm above. 
The offset-list representation of their union is computed by a sequence of merging, as
$
\merge(\ol_1,\merge(\ol_2,\dots\merge(\ol_{n-1},\ol_n)\ldots)).
$
Particularly, given two pointers to offset-lists $\ol,\ol'$, $\merge(\ol,\ol')$ implements their union: it chooses the offset-list that represents a set with the larger maximum, assume that it is $\ol$, and inserts the elements represented by the other list, $\ol'$, to it. We say that $\ol'$ \emph{is merged into} $\ol$.
This is done by the standard sorted-list merging in time 
$\bigo{|\ol'|}$ where $|\ol'|$ is the length of $\ol'$. 
Since $\ol'$ is without duplicities and with minimum 0,   
$\bigo{|\ol'|} \subseteq \bigo{\max(\ol')}$ where $\max(\ol')$ is the maximal element.



The $\bigo{\max(\ol')}$ cost is amortized by the cost of evaluating increments. The offset-list pair at $\ol'$ has seen at least $\max(\ol')-1$ increments since the only elements inserted into it are $0$, $1$, or, during merge, elements from other sets smaller than $\max(\ol')$. 
These increments of $\ol'$ are the budget used to pay for the mergeing of $\ol'$ into $\ol$. 
After the merge, the offset-list pair of $\ol'$ is discarded (as the CSA is non-replicating, it is no longer needed) hence 
the budget is used only once. 
%
Last, the assignment of the union to $c$ is done by a constant time assignment of a pointer to the \mbox{offset-list returned by the merge.}

\paragraph{Overall complexity of the simulation.}
Let us define the cost $\cost(x)$ of manipulations with the counter $x\in X$ during one step of the simulation as the sum of the costs of:
(1) evaluating all r-terms containing $c$,
(2) merging their offset-list into other ones, 
(3) creating offset-lists for terms $0$ or $1$ in $u(x)$ and merging them into other offset-lists,  
(4) the assignment of the result of $u(x)$ to $x$.  
The cost of processing a single letter $a_i$ is then the sum $\sum_{x\in X} \cost(x)$ and $|w|\cdot\sum_{x\in X} \cost(x)$ is the cost of the entire simulation.
Since the CSA is non-replicating and evaluating a single r-term is amortized constant time, the cost of (1) is in amortized constant time.
The cost of (2) is amortized by increments from step (1).
The creation and insertion of singletons in (3), at most two in $u(x)$, is constant time.
The pointer assignment in (4) is constant time.
The $\cost(x)$ is therefore amortized constant time, the amortized time of evaluating the update $u$ is in $\bigo{|X|}$, and the cost of the updates through the simulation is in $\bigo{|X|\cdot|w|}$. 
%
The cost of choosing the transitions, by evaluating their guards, is in $\bigo{|A|\cdot|w|}$ by the above analysis. The cost of testing the accepting condition at the reached configuration is in $\bigo{|A|}$ by an analogous argument. 

\begin{theorem}
If $A$ is non-replicating, then its simulation on $w$ takes $\bigo{|A|\cdot|w|}$ time.
\end{theorem}


\section{Augmented Determinization}
\label{sec:augmented}

In this section, 
we augment the subset construction from \cref{sec:determinization} with optimizations that prevent counter replication and hence extend the class of regexes that can be matched fast by simulation of the CSA.
It optimizations are tailored 
to CA with the special properties of $\caof R$, for a regex $R$, listed in \cref{sec:ca}.
%

%
%

\begin{figure}[t]
  \centering
  \begin{subfigure}[b]{0.33\textwidth}
    \centering
    \hspace{-0.5cm}%
    \begin{tikzpicture}[shorten >=1pt,node distance=1.8cm,on grid,initial text={},state/.style={circle, draw, minimum size=0.3cm}]
      \node[state] (q)              {\scriptsize$q$};
      \node[state] (r) [left=of q]  {\scriptsize$r$};
      \node[state] (s) [right=of q] {\scriptsize$s$};
      \node (h) at (-2.15,0.35) {\footnotesize\bf a)};
      \path[->] (q.165) edge node [above] {\scriptsize$a; x \ass x+1$} (r.15)
                (q.15) edge node [above] {\scriptsize$a; x \ass x+1$} (s.165)
                (s.-165) edge node [below] {\scriptsize$b; x \ass x$}   (q.-15)
                (r.-15) edge node [below] {\scriptsize$b; x \ass x$}   (q.-165);
    \end{tikzpicture}
  \end{subfigure}%
  \hfill
  \begin{subfigure}[b]{0.33\textwidth}
    \centering
    \begin{tikzpicture}[shorten >=1pt,node distance=1.8cm,on grid,initial text={},state/.style={circle, draw, minimum size=0.3cm}]
      \node[state] (q)               {\scriptsize$q$};
      \node[state] (r) [right=of q]  {\scriptsize$r$};
      \node (h) at (-1,0.35) {\footnotesize\bf b)};
      \path[->] (q.15) edge node [above] {\scriptsize$a; x \ass 1$} (r.165)
                (q) edge [loop left] node [left] {\scriptsize$a;$} node [below] {\scriptsize$x \ass x$} (q)
                (r) edge [loop right] node [right] {\scriptsize$a;$} node [below] {\scriptsize\hspace{1.5em}$x \ass x+1$}   (r)
                (r.-165) edge node [below] {\scriptsize$a; x \ass x+1$}   (q.-15);
    \end{tikzpicture}
  \end{subfigure}%
  \hfill
  \begin{subfigure}[b]{0.33\textwidth}
    \centering
    \begin{tikzpicture}[shorten >=1pt,node distance=1.8cm,on grid,initial text={},state/.style={circle, draw, minimum size=0.3cm}]
      \node[state] (q)              {\scriptsize$q$};
      \node[state] (r) [left=of q]  {\scriptsize$r$};
      \node[state] (s) [right=of q] {\scriptsize$s$};
      \node (h) at (-2.15,0.35) {\footnotesize\bf c)};
      \path[->] (q.165) edge node [above] {\scriptsize$a; x \ass x+1$} (r.15)
                (q.15) edge node [above] {\scriptsize$a; x \ass x$} (s.165)
                (s.-165) edge node [below] {\scriptsize$b; x \ass x+1$}   (q.-15)
                (r.-15) edge node [below] {\scriptsize$b; x \ass x$}   (q.-165);
    \end{tikzpicture}
  \end{subfigure}%
  \caption{Sub-structures of CA that are sources of counter replication.}
  \label{fig:replication}
\end{figure}
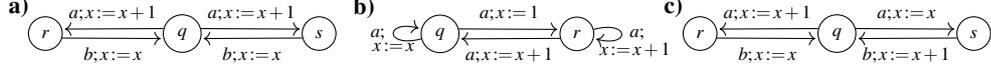

\paragraph{Intuition for the optimizations.}
\begin{sloppypar}
The emergence of counter replication and means of its elimination in the augmented construction, by techniques of \emph{counter sharing} and \emph{increment postponing},  are illustrated on simplified fragments of CA in \cref{fig:replication}.

In a),
$\scaut{\caof R}$ has transitions 
$\move{\move {\{q\}} {a,x_r\ass x_q+1,x_s\ass x_q+1} \{r,s\}} {b,x_q:=x_r \cup x_s} {\{q\}}$.  
%
%
%
The first transition replicates the entire content of the $x_q$, the second one unites the two sets. Both transitions are expensive.
The can be optimized by detecting that the values of $x_s$ and $x_r$ are the same, being generated by \emph{syntactically identical} updates, and storing the values in a \emph{shared counter} $x_{\{s,r\}}$.
This would result in transitions $\move{\move {\{q\}} {a,x_{\{r,s\}}\ass x_{\{q\}}+1} \{s,t\}} {b,x_{\{q\}}:= x_{\{r,s\}}} {\{q\}}$, with the replication and union eliminated.

Figure b) 
then illustrates why a counter $x_P$, $P\subseteq Q$, represents the set of values shared between the original counters $x_p$, $p\in P$.
That is, $x_P$ does not always hold the entire sets stored in the counters $x_p,p\in P$.  If their values are not the same, it stores only their intersection. The value of each $x_p$ is then partitioned among several shared counters $x_S$ with $p\in S$.
In b), 
$\scaut{\caof R}$ has transitions $\move{q} {a,x_q\ass x_q;x_r\ass 1} {\move{\{q,r\}} {a,x_q\ass x_q \cup x_r+1;x_r \ass 1\cup x_r+1} \{q,r\}}$, replicating the counter $x_r$.
Counter sharing would then generate transitions $\move{q} {a,x_{\{q\}}\ass x_{\{q\}};x_{\{r\}}\ass 1} {\move{\{q,r\}} {a,x_{\{q\}} \ass x_{\{q\}};x_{\{r\}} \ass 1;x_{\{q,r\}}\ass x_{\{r\}}+1} {\{q,r\}}}$ with counters $x_{\{q\}}$, $x_{\{r\}}$ for the subsets exclusive to $x_q$ and $x_r$, respectively, and $x_{\{q,r\}}$ for the intersection.

Last, in c), 
we illustrate the technique of \emph{increment postponing}. $\scaut{\caof R}$ would have transitions $\move{\move {\{q\}} {a,x_r\ass x_q+1,x_s\ass x_q} \{s,t\}} {b,x_q:=x_r \cup x_s+1} {\{q\}}$.
Since the increments on the two branches happen in different moments, the values of $x_r$ and $x_s$ differ  until the last increment of $x_s$ synchronizes them.
We avoid replication by storing the non-incremented value, obtained from $x_q$, in a counter shared by $x_r$ and $x_s$ and remembering that an increment of $x_r$ has been postponed. 
This is marked with $\plusmark$ in the name of the shared counter $x_{\{\marked{r},s\}}$. When the values of $x_r$ and $x_s$ synchronize (the increment is applied to $x_s$ too), the postponed increment is evaluated and the $\plusmark$-mark is removed. We would create transitions
$\move{\move {\{q\}} {a,x_{\{\marked r,s\}}\ass x_{\{q\}}} \{s,t\}} {b,x_{\{q\}}:= x_{\{\marked r,s\}}+1} {\{q\}}$.
If, before the synchronization, the value of the marked counter is either tested or incremented for the second time, we declare an \emph{irresolvable replication} and abort the entire construction (we allow postponing of only one increment). To prevent this situation from arising needlessly, we let states remember the counters that must have the empty value and we ignore these counters. 
\end{sloppypar}

\paragraph{Augmented Determinization Algorithm.}
\label{sec:algo2}
The augmented determinization produces from $\caof R = (X, Q, \Delta, \{\alpha_0\}, F)$ the CSA $\csaiiof {\caof R} = (\ii X, \ii Q, \ii \Delta, \{\ii \alpha_0\}, \ii F)$.
Its counters in $\ii X$ are of the form $x_S$ where $x\in X$ and $S\subseteq \marked Q$ and $\marked Q = Q \cup \set{\marked{q}\mid q\in Q}$.  
The guiding principle of the algorithm is that an assignment $\ii \smem$ of $\ii X$ represents an assignment $\smem$ of the counters in $\sc X$ of $\scaut {\caof R}$, namely, 
for each $x_q\in\sc X$, 
\begin{equation}
\smem(x_q) = \bigcup\nolimits_{q\in S,S\subseteq \marked Q} \ii\smem(x_S) \cup
\bigcup\nolimits_{\marked q\in S,S\subseteq \marked Q} \set{n + 1 | n \in \ii\smem(x_S)}.
\label{eq:cntmeaning}
\end{equation}

We will use some simplifying notation. 
As discussed in \cref{sec:ca}, by the construction of $\caof{R}$, the increment of $c$ and the guard  $x < \mxof x$ always appear on its transitions together, without any other guard on $x$.
Hence, in $\scaut{\caof R}$, all terms with an increment or filtering are of the form $\filterof {x < \mxof x}(x_\qq) + 1$. We will denote them by the shorthand $\finc {x_\qq}$ (we are using $\qq$ to denote an element from the set $\marked Q$, either $q$ or $\marked q$, for $q\in Q$).  

The states of $\csaiiof {\caof R}$ will additionally be distinguished according to which of the counters of $\ii X$ are \emph{active}, i.e., could have a non-empty value. 
Counters always valued by $\emptyset$ can be ignored, which simplifies transitions and decreases the chance of an irresolvable counter replication.
The states of $\csaiiof {\caof R}$ are thus of the form 
$(R, \ac)$ where $R\subseteq Q$ and $\ac\subseteq\ii X$ is a set of active counters.

The initial configuration is $\ii\conf_0 = ((\{q_0\},\{x_{\{q_0\}}\mid x\in X\}),\ii \smem_0)$ where $\ii\smem_0$ assigns $\{0\}$ to every $x_{\{q_0\}},x\in X$ and $\emptyset$ to every other counter in $\ii X$. 
The final condition assignment  
$\ii F((R,\ac))$ is, for each $(R,\ac)\in \ii Q$, 
constructed from $\sc F(R)$ 
by replacing every predicate $p[x_q]$ by the disjunction 
$p[x_q]^{\ac} = \bigvee_{x_S\in \ac,q\in S} p[x_S]$ 
that encodes $p[x_q]$ using the counters of $\ac$ in the sense of (\ref{eq:cntmeaning}).

The transitions in $\ii{\Delta\!}$ are constructed from transitions in $\sc {\Delta\!}$.
For source state $(R,\ac)\in \ii Q$, an original transition $\move{R}{a,\varphi,u}{R'}\in \sc {\Delta\!}$,
and set of active counters $\ac\subseteq \ii X$, $\ii{\Delta\!}$ has the transition 
$\move{(R,\ac)}{a,\ii\varphi,\ii u}{(R',\ac')}$, constructed as follows:

The guard $\ii \varphi$ is made from $\varphi$ by replacing every predicate $p[x_q]$ by the equivalent version with shared counters
$p[x_q]^{\ac}$ (as when constructing $\ii F$ above).

The update $\ii u$ is constructed in three steps. 
%
First, 
the update $\sharedu$ is made from $u$ by expressing the r-terms of $u$ using the shared counters $\ii X$. Each $t[x_q]$ is replaced by 
$$
\ii t = \bigcup
\Bigl(\,\bigl\{\, t[x_S]  \mid  x_S\in \ac, q\in S  \, \bigr\} \cup
\bigl\{\,\finc{t[x_S]}  \mid  x_S\in \ac,\marked q\in S  \, \bigr\}\,\Bigr) \ .
$$
Notice that all postponed increments are \emph{evaluated} in $\sharedu$, transformed to normal increments. 
If $\sharedu$ has an r-term $\finc{\finc{t}}$, i.e., a double increment, then the whole construction aborts and declares an \emph{irresolvable counter replication}. We allow postponing only one increment.%
\footnote{
Also transition guards and final conditions of $\csaiiof{\caof R}$ must not contain the $\plusmark$-mark since evaluating them regardless the postponed increments would return incorrect results. However, declaring counter replication on seeing a double increment here covers these cases \mbox{due to the structural properties of $\caof R$. }
}
Otherwise, we proceed to resolve counter replication. 
First, we make sure that every counter appears in the image of the update only in one kind of r-term.
We collect the set $\conflict$ of all r-terms $\finc {x_S}$ of $\sharedu$ with \emph{conflicting increments}, i.e. such that also $x_S$ is an r-term of $\sharedu$.
In update $\marked u$, conflicting increments are \emph{postponed}. For $x\in X$, $q\in Q$, and $\sharedu(x_q) = \bigcup T$,
$$
\marked u(x_q) =  \bigcup\bigl(\, T \setminus \conflict\,\bigr) 
\text{\ \ and\ \ }
\marked u(x_{\marked q})  =  \bigcup \bigl\{\,x_S \mid \finc{x_S} \in T \cap \conflict\,\bigr\} \ .
$$
The final update $\ii u$ then resolves counter replication, by grouping r-terms replicated in $\marked u$ under a common l-value (we call $z$ an \emph{l-value} of r-terms of $\marked u (z)$).
For an r-term $t$ of $\marked u$, let $\lval(t)$ be the set of its l-values.
Note that $\lval(t)$ is always of the form $\{x_\qq\}_{x\in S}$ for some fixed $x\in X$ (see property 4 of $\caof R$ in \cref{sec:ca}). 
We let $\ac'$ be the set of counters $x_S$ with $\lval(t) = \{x_\qq\}_{x\in S}$ for some r-term of $\marked u$. For all $x_S\in \ii X$, if $x_S\not\in \ac'$ then $\ii u(x_S) = \emptyset$ else 
$$
\ii u(x_S) = \bigcup\bigl\{\, t \mid t \text{ is an r-term of $\marked u$ and } \lval(t) = \{x_\qq\}_{\qq\in S}\,\bigr\} \ .  
$$

\begin{example}
\label{example:augmented}
Let us have 
$\move{R}{a,\varphi,u}{R'}\in\sc{\Delta\!}$ created in \cref{example:determ} with $R =
\{p,q\}$, $R' = \{r,s\}$, $\varphi = x_p{<}n \land x_p{\geq} m$, 
and $u = \{
x_r \ass \finc{x_p},
x_s \ass x_q \cup 1
\}$.
Let $\ac = \{x_{\set{p,q}},x_{\set{p,\marked q}}\}$.
Then 
$\sharedu = \{
x_r \ass \finc{x_{\set{p,\marked q}}} \cup \finc{x_{\set{p,q}}},
x_s \ass \finc{x_{\set{p,\marked q}}} \cup x_{\set{p,q}}  \cup 1
\}$.
Note that the $x_q$ in $u(x_s)$ becomes $\finc{x_{\set{p,\marked q}}}$, corresponding to the right part of the definition of $\ii t$ (the postponed increment $x_{\marked q}$ is evaluated in $\sharedu$).
Note that the r-term $\finc{x_{\set{p,q}}}$ is in $\conflict$ as $x_{\set{p,q}}$ is an r-term  of $\sharedu$ too.
Therefore it is postponed in $\marked u$, i.e. $\sharedu(x_r) = \finc{x_{\set{p,q}}} \cup \cdots$ becomes $\marked u(x_{\marked r}) = x_{\set{p,q}}$. We get
$\marked u = \{
x_r \ass \finc{x_{\set{p,\marked q}}},                          
x_s \ass \finc{x_{\set{p,\marked q}}} \cup x_{\set{p,q}} \cup 1,
x_{\marked r} \ass x_{\set{p,q}}                               
\}$.
Finally, $\ii u$ 
groups r-terms replicated in $\marked u$ under a common l-value:
$\ii u = \{
x_{\set{r,s}} \ass \finc{x_{\set{p,\marked q}}},
x_{\set{s}} \ass 1,
x_{\set{s,\marked r}} \ass x_{\set{p,q}}
\}$.
The next active counters are $\ac' = \{x_{\set{r,s}},x_{\set{s}},x_{\set{s,\marked r}}\}$.
%
Note that, for $x_{\set{p,\marked q}}$, the postponed increment at $\marked p$ was synchronized on this transition, while the conflict at $x_{\set{p,q}}$ was solved by postponing increment and marking $r$ with $\marked{}$.
\qed
\end{example}

The algorithm either returns the CSA $\csaiiof{\caof A}$, or detects an irresolvable counter replication, in which case $\csaiiof{\caof A}$ does not exist.\footnote{Aborting the construction here simplifies the description, but it would also be possible to continue the construction and return a DCSA that does not guarantee fast simulation.}
Let $m = \sharp R$ and recall that $n$ denotes the length of the matched text, $|w|$. 
Since $\caof R$ has at most $m$ states and $m^2$ transitions,
a basic analysis of the algorithm's data structures reveals that the resulting CSA has at most $2^{2^m}$ states, each with at most $2^{m^2}$ outgoing transitions, each transition of the size in $O(m2^m)$.
Because $\csaiiof{\caof A}$ encodes $\scaut {\caof A}$, it has the same language, and it also inherits its determinism. Since it does not replicate counters, it can be simulated in pattern matching fast, in time linear to the text and independent of the counter bounds. The following theorem is proved in \cref{app:complexity}.



\begin{theorem}
For $R$ with flat counting, if $\csaiiof {\caof R}$ exists, then 
it does not replicate counters, 
its size is in
$O(2^{2^m} m)$, 
$L(\caof R) = L(\csaiiof {\caof R})$,
and it can be simulated on a word $w$ of the length $n$ in time 
$O(2^{2m} m n)$.
\label{theo:complexity}
\end{theorem}

Matching can be done in time of constructing the CSA plus its simulation, which in the sum is indeed fast, not dependent on $k$ and linear in $n$. It can also be noted that the $m$ in the exponents above is not the size of the entire regex, but only the size of the counted sub-regexes.

\section{Regexes with Synchronizing Counting}
\label{sec:class}


Finally, in this section we define the class of regexes with synchronizing counting, which precisely captures when the CSA created by our construction in \cref{sec:augmented} does not replicate counters and hence allow fast matching (in the sense of \cref{theo:complexity}).

\begin{definition}[Regexes with synchronizing counting]
A regex has \emph{synchronizing counting} iff
it has no sub-expression $S\regex{\{n,m\}}$ 
where for some $k\in \nat$, a word from $L(S)^k$ has a prefix from $L(S)^{k+1}$.
\end{definition}

For instance, \regex{(ac*)\{1,4\}(ab|ba)\{3,5\}(a(ab)*)\{2,8\}} is a regex with synchronizing counting as each word from $L(\regex{ac*})^k$ must contain the symbol $a$ exactly $k$ times%
, words from $L(\regex{ab|ba})^k$ must have exactly $2k$ symbols, and words from $L(\regex{a(ab)*})^k$ can be uniquely split at the first $a$ in the \regex{a(ab)*}.
In comparison, \regex{(a|aa)\{2,5\}} does not have synchronizing counting as $a\cdot a \cdot a$ is a prefix of $aa \cdot aa$.

Intuitively, there is no pair of paths through $\caof{S\regex{\{m,n\}}}$ starting at the same state, over the same word, ending in the same state, where the number of increments differs by two. 
In such case, $\csaiiof {\caof {S\regex{\{m,n\}}}}$ would have to delay two increments, 
which our construction does not allow. 
The theorem below is proved in \cref{app:delayed}.

\begin{restatable}{theorem}{classtheorem}
\label{theo:class}
Given a regex $R$ with flat counting, 
the algorithm of \cref{sec:augmented} returns $\csaiiof{\caof R}$ if and only if $R$ has synchronizing counting.
\end{restatable}

\begin{corollary}
  Regexes with flat synchronizing counting have a fast matching algorithm.
\end{corollary}
\begin{proof}
  From \cref{theo:complexity,theo:class}.
\end{proof}

\paragraph{Counting with Markers.}
Even though designing and recognizing synchronizing counting is usually intuitive, 
it may also be tricky.  
For instance, $\regex{(\textbackslash \textbackslash \textbackslash \textbackslash d+\textbackslash \textbackslash \textbackslash \textbackslash .)\{3\}}$, 
from the database of real-world regexes we use in our experiment, has synchronizing counting, 
while
$\regex{ICE\_Dims.\{92\}((\_?(X|\textbackslash d+))\{13\})}$ does not.%
\footnote{
An automated way of identifying synchronizing counting would be running the CSA-to-DCSA determinization from \cref{sec:augmented}, but this is exponential to $|R|$.
}
A vast majority of real-world regexes we examined fortunately belong to very easily recognizable subclasses of synchronizing counting.
The most wide-spread and easy to recognize are regexes with \emph{letter-marked counting},
where every sub-expression $S\regex{\{m,n\}}$ has a set of marker letters such that every word from $L(S)$ has exactly one occurrence of a marker letter. \footnote{That letter-marked counting is a strict superset of the class that is in  \cite{oopsla} conjectured as handled by the algorithm of \cite{oopsla}. The conjecture of \cite{oopsla} is also not correct, as shown in \cref{app:counterexample}.}

Marker letters may be generalized to \emph{marker words}, though, 
markers that can arise by concatenation of several words from $L(S)$ cannot be used. The condition that has to be satisfied is that any word from $L(S)^k$, $k\in \nat$, has exactly $k$ non-overlapping occurrences of marker words as infixes. 
Another sufficient property of $S$ is that it has words of a \emph{uniform length}. 
The idea of markers may be generalized further until the point when the set of marker words is specified by general regexes, when we get precisely the synchronizing counting.
The regexes with letter-marked counting are easily human as well as machine recognizable (see a simple $O(|R|^2)$-time algorithm in \cref{sec:checker}).

\section{Practical Considerations}
\label{sec:practical}

Although the main point of this work is the theoretical feasibility of fast matching with synchronizing counting,
we will also argue that the results are of practical relevance. 
To this end, we show experimentally that synchronizing counting and marked
counting cover a majority of practical regexes. 
We also give arguments that matching with the CSA constructed in \cref{sec:augmented} can be done efficiently. 

\subsection{Occurrence of Synchronizing Counting in Practice}

To substantiate the practical relevance of synchronizing counting regexes, we examined a large sample of practical regexes using a simple checker of letter-marked counting (see \cref{sec:checker}). 
The benchmark consists of over 540\,000 regexes collected from 
(1) a large scale analysis of software
projects~\cite{DavisMCSL19}; 
(2) regexes used by network intrusion detection systems Snort~\cite{snort},
Bro~\cite{bro},
Sagan~\cite{sagan}, and the academic
papers~\cite{yang2010,tacas18-appred};
(4) the RegExLib database of regexes~\cite{regexlib}.
%


%
%
From the regexes that we could parse\footnote{We did not parse 38\,558 regexes since their syntax was broken or contained some advanced features we do not support.}, 31\,975 contained counting. 
We selected those with flat counting and with the sum of upper bounds of counters larger than 20 (as was done in \cite{oopsla} to filter out counting with small bounds that can be handled through counter unfolding and traditional methods)\footnote{926 regexes contain nested counting and 25297 regexes contain small upper bounds.}. 
This left us with 5\,751 regexes.
%
From these, only 46 regexes ($0.8\,\%$) have counting that is not letter-marked.
Furthermore, we manually checked these regexes and we identified that 22 of them have synchronizing counting.
We have therefore found only 24 regexes with non-synchronizing counting, i.e., 0.4\,\% of the examined set of regexes with flat counting. 

The 24 non-synchronizing regexes are listed in \cref{sec:nonvisible}. Some of
them may clearly be rewritten with synchronizing counting, such as
$\regex{(.+)\{25\}(.*)}$, which can be rewritten as $\regex{.\{25,\}(.*)}$. We
speculate that some of them might in fact represent a mistake, such as
$\regex{(.*)\{1,32000\}[bc]}$ where the counter matches the empty word, or
$\regex{(\textbackslash n\textbackslash s+)(criterion .*\textbackslash
n)(\textbackslash s.+)\{1,99\}}$ where the $\regex{\textbackslash s.+}$ might
have been intended as $\regex{\textbackslash s\textbackslash S+}$
($\regex{\textbackslash s}$ are white spaces, $\regex{\textbackslash S}$ are all
the other characters). Synchronizing counting seems to capture the intuition with which counting is often written, hence reporting non-synchronizing counting might help identifying bugs.  

By the same methodology and from a nearly identical benchmark, \cite{oopsla} arrived to a sample of 5\,000 regexes with flat counting with the sum of bounds larger than 20.
The algorithm of \cite{oopsla} did not cover 571 regexes from the 5\,000, which is
11\,\% of the examined set of regexes with flat counting (in contrast to the
$0.4\,\%$ with non-synchronizing counting and the $0.8\,\%$ with counting that is not letter-marked, measured on a slightly larger set of regexes).
The two sets of regexes with flat counting, the 5\,751 of ours and the 5\,000 of \cite{oopsla}, are not perfectly identical, however. Differences are to a small degree caused by differences in the base database
(\cite{oopsla} uses about 18 more regexes that are proprietary and excludes 26
regexes with counter bounds larger than 1\,000), and to a larger degree by small differences in the parsers.

\subsection{Practical Efficiency of Matching with Synchronizing Counting}
The size and the worst-case time of simulation of $\csaiiof{\caof R}$ are still
exponential to the number of states of $\caof R$ (namely, $O(
2^{2^m} m)$ and $O(2^{2m} m n)$ where $m=\sharp R$ equals the number of states of $\caof R$, cf. \cref{theo:complexity}). The potential problem is that the algorithm may generate at most $2^m$ counters, and this potentially threatens practicality of our matching algorithm. 

First, it should be noted that the $m$ in the exponent can be decreased from the size of the entire regex to the size of the counted sub-expression, which is usually very small. 
Then, although an efficient implementation is beyond the scope of this paper and we are leaving it as a future work, we give some indirect arguments for practicality of the CA-to-CSA algorithm.%
\footnote{A competitive matcher that runs on real-world regexes requires an extensive infrastructure, optimized data structures for the shared registers, and ideally an on-the-fly version of the CA-to-CSA determinization (similar to the online DFA simulation).}

By the standard techniques of register allocation \cite{dragoonBook}, 
it is possible to decrease the number of counters and counter assignments other than identity dramatically.
In fact, simply eliminating needless renaming of counters and reusing the same name whenever possible,   
our algorithm creates CSA isomorphic to those of \cite{oopsla} when run on regexes handled by \cite{oopsla}. 
The work \cite{oopsla} already shows that simulating these CSA may be done
efficiently and that it brings dramatic improvements over best matchers on
counting-intensive examples. 

In our experience with hand-simulating the algorithm on practical examples,
cases not handled by \cite{oopsla} do not behave much differently, and the numbers of CSA counters do not have a strong tendency to explode.

\section{Conclusions}
We have extended the regex matching algorithm of \cite{oopsla} and shown that
the extended version allows fast pattern matching of so-called synchronising
regexes, a class of regexes that we have newly introduced.
The class of synchronising regexes significantly extends all previously known
classes of regexes that allow fast matching and covers a majority of regexes
appearing in practice (wrt. our empirical study).

In the future, we plan to study extensions of the presented techniques to regexes with nested counting (non-flat). This will probably require a more sophisticated alternative of the offset-list data structure for sets, capable of storing relations of numbers. 
An interesting question is also how and when regexes can be rewritten to a synchronizing form and for what cost. 

\section*{Acknowledgment}
This work has been supported by the Czech Ministry of Education, Youth and
Sports project LL1908 of the ERC.CZ programme, the Czech Science Foundation
project 23-06506S, and the FIT BUT internal project FIT-S-23-8151.

\bibliography{literature_usenix.bib,bibliography.bib}

\newpage
\eject
\appendix
\section{Proof of \cref{theo:alg1corectness}}
\label{app:alg1correctnessproof}

\renewcommand{\csaof}[1]{\scaut{#1}}

\correctness*
\begin{proof}
  Let $A = (\vset, Q, \Delta, I, F)$. $\scaut{A}$ is deterministic since it has a single initial configuration and the guards of transitions originating in the same state are minterms.

  The size of $|\scaut{A}| \in \bigo{2^|Q| + 2^|Delta|.|Delta|} = \bigo{|A|}$ where $2^|Q|$ is the number of states and $2^|Delta|.|Delta|$ is the number of transitions times their size.
  
  As $L(A) = L(\confaut{A})  = L(\dfaof{\confaut{A}})$ and $L(\csaof{A}) = L(\confaut{\csaof{A}})$ we need to only show that $L(\dfaof{\confaut{A}}) = L(\confaut{\csaof{A}})$.
  We do this by showing that $\dfaof{\confaut{A}}$ and $\confaut{\csaof{A}}$ are isomorphic.

  We show that $\encode$ is the bijection that gives the isomorphism.
  Because $\dfaof{\confaut{A}}$ has only Cartesian states, $\encode$ is injective.
  To show that $\encode$ is the bijection that gives the isomorphism we need to show that
  \begin{enumerate}[(1)]
    \item as $I$ is the initial state of $\dfaof{\confaut{A}}$, $\encode(I)$ is the initial state of $\confaut{\csaof{A}}$,\label{enum:iso:init}
    \item for each state $\alpha$ of $\dfaof{\confaut{A}}$, $\encode(\alpha)$ is a state of $\confaut{\csaof{A}}$,\label{enum:iso:encstate}
    \item for each state $\beta$ of $\confaut{\csaof{A}}$, $\decode(\beta)$ is a state of $\dfaof{\confaut{A}}$, i.e. $\encode$ is surjective,\label{enum:iso:decstate}
    \item if $\move{\alpha}{a}{\alpha'}$ is a transition of $\dfaof{\confaut{A}}$ then there is a transition $\move{\encode(\alpha)}{a}{\encode(\alpha')}$ in $\confaut{\csaof{A}}$,\label{enum:iso:enctrans}
    \item if $\move{\beta}{a}{\beta'}$ is a transition of $\confaut{\csaof{A}}$ then there is transition $\move{\decode(\beta)}{a}{\decode(\beta')}$ in $\dfaof{\confaut{A}}$,\label{enum:iso:dectrans}
    \item $\alpha$ is a final state of $\dfaof{\confaut{A}}$ iff $\encode(\alpha)$ is a final state of $\confaut{\csaof{A}}$.\label{enum:iso:final}
  \end{enumerate}

  \ref{enum:iso:init} holds trivially from the construction.

  \begin{sloppypar}
  We now assume that this \emph{Argument} holds: for each transition $\move{\alpha}{a}{\alpha'}$ of $\dfaof{\confaut{A}}$, if $\encode(\alpha)$ is a state of $\confaut{\csaof{A}}$, then there is a transition $\move{\encode(\alpha)}{a}{\encode(\alpha')}$ in $\confaut{\csaof{A}}$.
  If this holds, then we can easily show that \ref{enum:iso:encstate} holds, because if $\alpha$ is a state of $\dfaof{\confaut{A}}$, then there is a path $\move{\alpha_0}{a_1}{\alpha_1}\move{}{a_2}{\dots}\move{}{a_{n}}{\alpha}$ in $\dfaof{\confaut{A}}$ and by induction on the length of this path, we can show that $\move{\encode(\alpha_0)}{a_1}{\encode(\alpha_1)}\move{}{a_2}{\dots}\move{}{a_{n}}{\encode(\alpha)}$ is a path in $\confaut{\csaof{A}}$ which means that $\encode(\alpha)$ is a state of $\confaut{\csaof{A}}$. 
  From the \emph{Argument} and \ref{enum:iso:encstate} it immediately follows that \ref{enum:iso:enctrans} holds too.
  From
  \begin{itemize}
    \item \ref{enum:iso:enctrans},
    \item the fact that $\confaut{\csaof{A}}$ is deterministic,
    \item the fact that $\dfaof{\confaut{A}}$ is deterministic and total (i.e. for each state $\alpha$ and symbol $a$, there is exactly one transition $\move{\alpha}{a}{\alpha'}$),
    \item and the fact that $\encode(\decode(\beta)) = \beta$ for each state $\beta$ of $\confaut{\csaof{A}}$ 
  \end{itemize}
  it easily follows (again by using induction on the length of path to $\beta$) that \ref{enum:iso:decstate} and \ref{enum:iso:dectrans} hold too.
  \end{sloppypar}

  We now only need to prove that the \emph{Argument} holds.
  Let therefore $\move{\alpha}{a}{\alpha'}$ be a transition of $\dfaof{\confaut{A}}$.
  We need to show that if $\encode(\alpha)$ is a state of $\confaut{\csaof{A}}$, then there is a transition $\move{\encode(\alpha)}{a}{\encode(\alpha')}$ in $\confaut{\csaof{A}}$.
  Let $\encode(\alpha) = (R,\smem)$ and $\encode(\alpha') = (R',\smem')$. 
  Let $\Delta^{A}_{R,a} = \{\,\move{q}{a,\phi, u}{p} \in \Delta \mid q \in R\,\}$ be the set of transitions in RA $A$ starting in some state from $R$.
  The transition $\move{\alpha}{a}{\alpha'}$ is formed from the constituent transitions $\move{(q,\mem)}{a}{(q',\mem')}$ of $\confaut{A}$ where $(q,\mem) \in \alpha$ and $(q',\mem') \in \alpha'$.
  For each such transition $\move{(q,\mem)}{a}{(q',\mem')}$, there has to exist some transition $\move{q}{a,\phi,u}{q'}$ in $A$ where $\mem \models \phi$ and $\mem'(x) = u(x)(\mem)$ for each $x \in X$.
  From the assumption of the construction, we know that $\phi$ is a conjunction of some unary predicates $p[x]$ which means that $\mem \models p[x]$.
  From this we know that $\smem \models p[x_q]$ for each such $p[x]$, because $\mem(x) \in \smem(x_q)$ and  $\smem \models p[x_q]$ holds if there exists some value satisfying $p[x_q]$ in $\smem(x_q)$.
  This also means that $\smem \models \indexby \phi q$.
  
  Let therefore $\Delta^{\alpha}_{R,a} = \{\,\move{q}{a,\phi, u}{q'} \in \Delta^{A}_{R,a} \mid \mem \models \phi$ for some $(q,\mem) \in \alpha\,\}$ be the set of transitions in $A$ that create some transition in $\confaut{A}$ that are constituent for $\move{\alpha}{a}{\alpha'}$.
  Note that $\Delta^{\alpha}_{R,a}$ is a subset of $\Delta^{A}_{R,a}$ and that for each transition $\move{q}{a,\phi, u}{q'}$ in $\Delta^{A}_{R,a} \setminus \Delta^{\alpha}_{R,a}$ there is no $(q,\mem) \in \alpha$ such that $\mem \models \phi$.
  Using this, we can show that $\smem \not\models \indexby \phi  q$, i.e. $\smem \models \neg\indexby \phi q$.
  Assume that this is not true, i.e. $\smem \models \indexby \phi q$.
  Then for each atomic predicate $p[x]$ in $\phi$ there has to exist some value in $\smem(x_q)$ that satisfies $p[x_q]$.
  For each $x \in X$ there is at most one $p[x]$ in $\phi$ therefore we can take $\mem$ such that $\mem(x)$ is the value in $\smem(x_q)$ that satisfies $p[x]$ (if there is no $p[x]$, take some random value from $\smem(x_q)$).
  There must be a configuration $(q, \mem)$ in $\alpha$ as $\alpha$ is Cartesian.
  We then have $\mem \models \phi$ which is a contradiction with $\move{q}{a,\phi, u}{q'} \not\in \Delta^{\alpha}_{R,a}$.

  From the assumption that $\encode(\alpha) = (R, \smem)$ is the state of $\confaut{\csaof{A}}$, we know that $R$ must be a state of $\csaof{A}$, which means that we create some transitions from $R$ labelled with $a$.
  We have that $\encode(\alpha') = (R',\smem')$ and we now show that the algorithm creates a transition $\move{R}{a,\psi,u'}{R'}$ such that $(R',\smem')$ is the image of $(R,\smem)$ under this transition, i.e. $\move{\encode(\alpha)}{a}{\encode(\alpha')}$ is the transition of $\confaut{\csaof{A}}$. 
  Let $\Delta_{R, a}$ be the set of constituent $a$-transitions as defined in \cref{sec:determinization}, i.e. set-versions of transitions in $\Delta^{A}_{R,a}$.
  Let $\psi$ be a conjunction of
  \[\bigwedge_{\move{q}{a,\phi, u}{q'} \in \Delta^{\alpha}_{R,a}} \indexby \phi  q
  \quad\text{and}\quad
  \bigwedge_{\move{q}{a,\phi, u}{q'} \in \Delta^{A}_{R,a} \setminus \Delta^{\alpha}_{R,a}} \neg\indexby \phi q.\]
  As we have shown, $\smem \models \indexby \phi  q$ for each conjunct of the first big conjunct and $\smem \models \neg\indexby \phi  q$ for each conjunct of the second big conjunct, which means that $\smem \models \psi$.
  Furthermore, because guard of each transition in $\Delta_{R, a}$ occurs in $\psi$, \\$\psi \in \Minterms{\set{\phi | \move{q}{a,\phi, u}{p} \in \Delta_{R, a}}}$.
  Then the transition $\move{R}{a,\psi,u'}{R'}$ constructed for minterm $\psi$ is the one we need and we now show this.
  Let $\Delta_{R, a, \psi} \subseteq \Delta_{R, a}$ be the transitions with guards compatible with minterm $\psi$, i.e. they are the set-version of transitions from $\Delta^{\alpha}_{R,a}$.
  Obviously $R' = \{\,p \mid \move{q}{a,\phi, u}{p} \in \Delta_{R, a, \psi}\,\} = \{\,p \mid \move{q}{a,\phi, u}{p} \in \Delta^{\alpha}_{R,a}\,\}$ is the correct set of states of $\encode(\alpha')$.
  The only thing that is remaining to be shown is that applying the update $u'(x_q)$ on $\smem$ results in the correct value of $\smem'(x_q)$ for each $x \in X, q \in Q$.
  
  We fix $x \in X, q' \in Q$ and let $N$ be the set of values that we get from applying the update $u'(x_{q'})$ on $\smem$, i.e. $N = \bigcup \set{u(x_{q'})(\smem) | \move{q}{a,\phi,u}{q'} \in \Delta_{R, a, \psi}}$. 
  We need to show that $N = \smem'(x_{q'})$.

  Firstly, we show that if $n \in \smem'(x_{q'})$ then $n \in N$.
  Because $n \in \smem'(x_{q'})$ and $\encode(\alpha') = (R', \smem')$, there has to exist some $(q',\mem') \in \alpha'$ where $\mem'(x) = n$.
  For this, there has to be some constituent transition $\move{(q,\mem)}{a}{(q',\mem')} \in \confaut{A}$ of the transition $\move{\alpha}{a}{\alpha'}$ for some $(q,\mem) \in \alpha$.
  This implies that there is some transition $\move{q}{a,\phi,u}{q'} \in \Delta^{\alpha}_{R,a}$ where $\mem \models \phi$ and $\mem'(x) = u(x)(\mem)$.
  From this we know that the term $t$, which is constructed from $\indexby q {u(x)}$ by adding filters created from $\phi$, must be one that of the terms used for computing $N$.
  We also know that for each $y \in X$ the value $\mem(y) \in \smem(y_q)$ and because $\mem \models \phi$ these values will not be filtered out by the filter in $t$.
  Therefore, $n$ must be in $N$.
  
  Let now $n \in N$ and we show that $n \in \smem'(x_{q'})$.
  From $n \in N$ we know that there must be some $\delta = \move{q}{a,\phi,u}{q'} \in \Delta_{R, a, \psi}$ such that $n \in u(x_{q'})(\smem)$.
  We can create $\mem$ such that the values of $\mem$ are exactly the ones that result in $n$ after applying $u(x_{q'})$.
  Furthermore, there is a transition $\delta'$ in $\Delta^{\alpha}_{R,a}$ from which $\delta$ was created and because $\alpha$ is Cartesian, $(q, \mem) \in \alpha$ (and it is allowed in $\delta'$ from the way filters work).
  This means that $n \in \smem'(c_{q'})$.


  Finally, we need to show that \ref{enum:iso:final} holds.
  Let $\alpha$ be a state of $\dfaof{\confaut{A}}$ and $\encode(\alpha) = (R, \smem)$.
  If $\alpha$ is final then there has to be $(q,\mem) \in \alpha$ where $\mem \models F(q)$. 
  By similar argument as we have shown for guards, $\smem \models \indexby{\phi} q$ which means that $\smem \models \sc{F}(R) = \bigvee_{q \in R} \indexby{F(q)} q$, i.e. $\encode(\alpha)$ is a final state of $\confaut{\csaof{A}}$.
   In the other direction, if $\encode(\alpha)$ is final, then $\smem \models \bigvee_{q \in R} \indexby{F(q)} q$ which means there has to be some $q \in R$ where $\smem \models \indexby{F(q)} q$.
  We can now create $\mem$, where for each $x \in X$ we have that the value $\mem(x)$ is some value of $\smem(x_q)$ that satisfies atomic predicate $p[x_q]$ from $\indexby{F(q)} q$ (there is at most one such atomic predicate and if it does not exist, we take random value of $\smem(x_q)$).
  Then $\mem \models F(q)$ and there must be $(q,\mem) \in \alpha$ as $\alpha$ is Cartesian, therefore $\alpha$ is a final state.%
\end{proof}

\section{Counting Regex to CA}
\label{sec:regextoca}

We show how to construct counting automaton $\caof R = (X, Q, \Delta, I, F)$ for a counting regex $R$ and give some useful properties of it.
Our construction is based on the generalization of the Glushkov construction~\cite{Glu61}
where we follow the construction from~\cite{Gelade09_countingregex} with slight modifications.

Glushkov construction creates an automaton whose states are occurrences of symbols in the regex $R$ and a special initial state.
We therefore define \emph{marked} regex $\markedregex{R}$ that is over the alphabet $\Sigma \times \nat$ such that no symbol from this alphabet occurs twice in $\markedregex{R}$.
As an example, for regex $\regex{abc(bc)\{3\}}$, the marked version could be $\regex{a}_1\regex{b}_1\regex{c}_1(\regex{b}_2\regex{c}_2)\regex{\{3\}}$.
The states of the automaton will then be $\regex{a}_1$, $\regex{b}_1$, $\regex{c}_1$, $\regex{b}_2$, and $\regex{c}_2$ and the special initial state.
We also write $a_i \in \markedregex{R}$ if $a_i$ occurs in $\markedregex{R}$.
Note that the number of these symbols in $\markedregex{R}$ is $\sharp R$.

Counters $X$ of $\caof R$ are created for each counting sub-expressions in $R$.
To simplify the construction, we assume that if $\epsilon \in L(\rexcount{S}{m,n})$ for some sub-expression $\rexcount{S}{m,n}$ of $R$, then $m = 0$ (transforming regex $R$ to this form can be done in linear time~\cite{Gelade09_countingregex}).
We also assume that $R$ does not contain any sub-expressions $\rexstar{S}$, instead it uses sub-expressions $S\{0,\infty\}$ for them.
However, $\caof R$ will not contain counters for these counting sub-expressions, as the only guards for such counters would be $c \geq 0$ and $c < \infty$ which are both trivially always true.
The set $X$ of counters therefore contains all counting sub-expressions of $\markedregex{R}$ other than the ones of the form $\rexcount{\markedregex{S}}{0,\infty}$. 
Furthermore, for each sub-expression $\markedregex{S}$ of $\markedregex{R}$ we define the set $\subctrs{\markedregex{S}}$ that contains all counting sub-expressions $\markedregex{T} \in X$ of $\markedregex{S}$ such that $\markedregex{T} \not= \markedregex{S}$.
We also define $\lowerctr{\rexcount{\markedregex{S}}{m,n}} = m$ and $\upperctr{\rexcount{\markedregex{S}}{m,n}} = n$.

We now define the sets $\first{\markedregex{R}}$ and $\last{\markedregex{R}}$ which consist of first (last) symbols in some word denoted by $\markedregex{R}$.
They are defined inductively as:
\begin{itemize}
  \item $\first{\rexeps} = \last{\rexeps} = \emptyset$ and $\first{\rexa_i} = \last{\rexa_i} = \{\rexa_i\}$
  \item $\first{\markedregex{R_1\regex{|}R_2}} = \first{\markedregex{R_1}} \cup \first{\markedregex{R_2}}$ and $\last{\markedregex{R_1\regex{|}R_2}} = \last{\markedregex{R_1}} \cup \last{\markedregex{R_2}}$
  \item if $\epsilon \in L(\markedregex{R_1})$, then $\first{\markedregex{R_1R_2}} = \first{\markedregex{R_1}} \cup \first{\markedregex{R_2}}$, else $\first{\markedregex{R_1R_2}} = \first{\markedregex{R_1}}$
  \item if $\epsilon \in L(\markedregex{R_2})$, then $\last{\markedregex{R_1R_2}} = \last{\markedregex{R_1}} \cup \last{\markedregex{R_2}}$, else $\last{\markedregex{R_1}\markedregex{R_2}} = \last{\markedregex{R_2}}$
  \item $\first{\rexcount{\markedregex{R}}{m,n}} = \first{\markedregex{R}}$ and $\last{\rexcount{\markedregex{R}}{m,n}} = \last{\markedregex{R}}$
\end{itemize}
Using these sets, we define the set $\follow{\markedregex{R}}$, which contains triples $(a_i, b_j, x)$ where $a,b \in \Sigma$, $i,j \in \nat$ and $x$ is either $\emptyctr$ or a counting sub-expressions of $\markedregex{R}$.
This set is used to create transitions in the counting automaton $\caof R$ between states $a_i$ and $b_j$ where $x$ gives us information on which counters to increment.
Formally, the set $\follow{\markedregex{R}}$ is a union of sets
\[\Set{(a_i, b_j, \emptyctr) | \markedregex{S_1}\markedregex{S_2}\text{ is a subexpression of }\markedregex{R}\text{, }a_i \in \last{\markedregex{S_1}}\text{, and }b_j \in \first{\markedregex{S_2}}}\]
and
\[\Set{(a_i, b_j, \rexcount{\markedregex{S}}{m,n}) | \rexcount{\markedregex{S}}{m,n}\text{ is a subexpression of }\markedregex{R}\text{, }a_i \in \last{\markedregex{S}}\text{, and }b_j \in \first{\markedregex{S}}}.\]


We can now define the counting automaton $\caof R = (X, Q, \Delta, I, F)$.
The set of states $Q$ contains the initial state $q_0$ and all symbols $a_i$ occurring in $\markedregex{R}$.
The set $\Delta$ contains following transitions:
\begin{itemize}
  \item for each $a_i \in \first{\markedregex{R}}$, there is transition $\move{q_0}{a, true, u}{a_i} \in \Delta$ where for each $\markedregex{S} \in X$
  \[ u(\markedregex{S}) = 
    \begin{cases}
      1 & \text{if } a_i \in \markedregex{S}\\
      0 & \text{else}
    \end{cases}
  \]
  \item for each $(a_i, b_j, \emptyctr) \in \follow{\markedregex{R}}$, there is a transition $\move{a_i}{b, \phi, u}{b_j} \in \Delta$ where
  \[\phi = \bigwedge_{\markedregex{S} \in X, a_i \in \last{\markedregex{S}}} \markedregex{S} \geq \lowerctr{\markedregex{S}}\]
  and for each $\markedregex{S} \in X$
  \[u(\markedregex{S}) =
    \begin{cases}
      1 & b_j \in \first{\markedregex{S}}\\
      \markedregex{S} & a_i,b_j \in \markedregex{S}\text{ and }b_j \not\in \first{\markedregex{S}}\\
      0 & \text{otherwise}
    \end{cases}
  \]
  \item for each $(a_i, b_j, \markedregex{S}) \in \follow{\markedregex{R}}$, there is a transition $\move{a_i}{b, \phi, u}{b_j} \in \Delta$ where either
  \[\phi = \markedregex{S} < \upperctr{\markedregex{S}} \land \bigwedge_{\markedregex{T} \in \subctrs{\markedregex{S}}, a_i \in \last{\markedregex{T}}} \markedregex{T} \geq \lowerctr{\markedregex{T}}\]
  for the case that $\markedregex{S} \in X$ or
  \[\phi = \bigwedge_{\markedregex{T} \in \subctrs{\markedregex{S}}, a_i \in \last{\markedregex{T}}} \markedregex{T} \geq \lowerctr{\markedregex{T}}\]
  otherwise (i.e. the case where $\lowerctr{\markedregex{S}} = 0$ and $\upperctr{\markedregex{S}} = \infty$) and for each $\markedregex{T} \in X$
  \[u(\markedregex{T}) =
    \begin{cases}
      \markedregex{S} + 1 & \markedregex{T} = \markedregex{S}\\
      1 & \markedregex{T} \in \subctrs{\markedregex{S}} \text{ and } b_j \in \first{\markedregex{T}}\\
      \markedregex{T} & \markedregex{T} \not\in \subctrs{\markedregex{S}}, \markedregex{T} \neq \markedregex{S} \text{ and } a_i,b_j \in \markedregex{T}\\
      0 & \text{otherwise}
    \end{cases}
  \]
\end{itemize}
The set of initial configurations is a singleton $\set{(q_0, \set{\markedregex{S} \mapsto 0 | \markedregex{S} \in X})}$.
The final condition $F$ is defined for every state $q$ as
\[F(q) = \begin{cases}
  true & q = q_0 \text{ and } \epsilon \in L(R)\\
  \bigwedge_{\markedregex{S} \in X, a_i \in \last{\markedregex{S}}} \markedregex{S} \geq \lowerctr{\markedregex{S}} & \text{if } q \in \last{\markedregex{R}}\\
  false & \text{otherwise}
\end{cases}\]

\begin{lemma}
  Let $R$ be a regular expression with $\sharp R = m$ occurrences of symbols. Then $L(R) = L(\caof{R})$ and $\caof{R}$ has $m+1$ states and up to $\cntof R.m^2$ transitions. If $R$ is flat, it has only up to $m^2$ transitions.
\label{lemma:glushkov}
\end{lemma}
\begin{proof}
  By induction on the structure of $R$.
\end{proof}

\subsection{Properties of $\caof R$ for Flat Regexes}
\label{sec:caproperties}
The reason for using Glushkov construction instead of Antimirov construction as in \cite{oopsla} is that our $\caof R$ has clearly delimited \emph{counting loops}, that is, the parts corresponding to the counting sub-expressions of $\markedregex{R}$ with special conditions on the entry and exit transitions of the counting loops.

Let us summarize properties of a counting loop for $x\in X$ for flat regex $R$. 
%
%
It is characterized by a set of states $\states x \subseteq Q$, with $\states x\cap\states y = \emptyset$ for $x\neq y$ 
a set of \emph{entry states} $\entry x \subseteq \states x$ 
and the set of \emph{exit states} $\last{x} \subseteq \states x$.
As $x$ is some counting sub-expression $\rexcount{\markedregex{S}}{m,n}$, the sets $\entry x$ and $\last{x}$ are defined as in previous section and the set $\states x$ contains all $a_i \in x$.
Transitions of $\caof R$ not incident with $\states x$ assign $0$ to $x$ and have no occurrence of $x$ in the guard.
The transitions incident with $\states x$ are partitioned into $\Delta_x =\trinner \uplus \trinc \uplus \trrestart \uplus \trin \uplus \trout$ where $\trrestart$ is optional:
\begin{itemize}
\item
$\trin$ are the only transitions from outside into $\states x$. 
$\trout$ are the only transitions from $\states x $ outside.
\item
Transition in $\trinner$ start and end in $\states x$. Their guard is $\top$ and the update is $c\ass x$. They form paths from the entry states to some state in $\last{x}$. 
\item
$\trinc$ has transitions $\move {q_o} {a,x<\mxof x,x \ass  x+1} q_e$, $q_o \in \last{x}$, such
that $\entry x$ is \emph{exactly} the set of all their targets $q_e$. 
$\trinc$ also determines the remaining transition sets.
\item
$\trrestart$ is obtained from $\trinc$ by replacing the guard by $x \geq \mnof x$ and replacing the $x+1$ in the update by $1$.
\item
$\trin$ are all transitions from outside into $\states x$.
It is a union of sets $\Delta_{q}$ for several states $q\in Q\setminus  \states x$.
%
Each $\Delta_q$ is obtained from $\trrestart$ by replacing the source $\head x$ by $q$ and removing $x < \mxof x$ from the guard.  
\item
$\trout$ are all transitions leaving $\states x$. They are of the form $\move {q_o} {a,x\geq \mnof x,x\ass 0} {q}$, $q_o \in \last{x}$, $q\in Q\setminus \states x$. 
These transitions may be shared with $\Delta^{\mathit{in}}_y$ in which case the update also has $y\ass 1$. 
\end{itemize}
No state in $\states x$ is initial, and states from $\last x$ are the only state in $\states
x$ that may have a satisfiable final condition, 
namely, it may be $x \geq \mnof x$ (or $\bot$). 
States outside $\states x$ do not contain $x$ in their final conditions.  

Below, we summarize characteristics of paths through $\caof R$ where a
\emph{path through a CA (or CSA)} from $r_0$ to $r_n$ over a word $a_1\cdots
a_n, n\in\nat$, is a sequence of transitions 
$\move{r_0}     {a_1,\phi_1,u_1} {r_1}$, 
$\move{r_1}     {a_2,\phi_2,u_2} {r_2}$, \ldots,  
$\move{r_{n-1}} {a_n,\phi_n,u_n} {r_n}$.
We define the \emph{$c$-fanout from $q\in Q$ over a word $v$} as the set of
paths from $q$ over $v$ that start by assigning 1 to $c$ (by $\trin$ or
$\trrestart$) and go on by copying or incrementing $c$ (by $\trinner$ or $\trinc$)~only.

\begin{lemma}
\label{lemma:caofr}
Let $\pi$ be a path in the $c$-fanout from $q$
split by transitions from $\trinc$ into $n+1$ fragments, $n\in \nat$ (new fragments start with the transition from $\trinc$). 
Then, assuming $c = \rexcount{\markedregex{S}}{m,n}$, each fragment reads a word in $L(S)$ or is a prefix of some word for the last fragment.
Furthermore, $\pi$ increments $c$ $n$-times and each fragment except the first one starts in some state from $\last{c}$.
\end{lemma}

\section{Complexity of Computing and Simulating $\csaiiof R$}\label{app:complexity}
Let $R$ be a regex with flat counting
and let $A = \caof R = (X, Q, \Delta, I, F)$.
If $X$ is non-empty (and in further we assume it is), then for a state $(R,\ac)$ of $\csaiiof{A}$, the set of states $R$ is unnecessary, it is possible to obtain it from the set of active counters $\ac$ (as every $x_r$, $r \in R$ must have some value, i.e. some $x_S$ with $r$ or $\marked r$ in $S$ must be active).
Each state of $\csaiiof{A}$ is then uniquely represented by some subset of $\ii X = X \times \pow{Q \cup \marked Q}$, i.e. $\csaiiof A$ has $\bigo{2^{|X|2^{|Q|}}}$ states. However, because $R$ has flat counting, for each $q \in Q$ there is at most one $x \in X$ such that $x_q$ can have non-zero value in $A$. Therefore, there will be at most $\bigo{2^{2^{|Q|}}}$ states in $\csaiiof A$.

For each state $(R, \ac)$ there are at most $O(2^{|\Delta|}) = \bigo{2^{|Q|^2}}$ transitions starting from it ($|\Delta|$ has at most $|Q|^2$ transitions for regexes with flat counting, see \cref{lemma:glushkov}), as we create transition for each minterm of $\scaut{A}$ which are conjunctions of guards of transitions from $A$.
Furthermore, the size of each transition depends on the size of its guard and the size of the update function.
The size of the guard (i.e. minterm) of transitions in $\scaut{A}$ is $\bigo{|Q|}$ as each such guard is conjunction of predicates $x_q < \mxof{x}$ and $x_q \geq \mnof{x}$, for each $q \in Q$ with its at most one nonzero counter $x \in X$, and their negations.
As we replace each $x_q$ in this guard by disjunction of $x_S$, $q \in S$, for $\csaiiof{A}$, the size of the guard in $\csaiiof{A}$ is $\bigo{|Q|2^{|Q|}}$.
As each $x_S$ can occur only in one term in the update (we cannot have counter replication), its size is $\bigo{|X|2^{|Q|}}$, however, with similar reasoning as for number of states, the size of the update can be reduced to $\bigo{2^{|Q|}}$.
The size of $\csaiiof{A}$ is then $\bigo{2^{2^{|Q|}} * 2^{|Q|^2} * |Q|2^{|Q|}} = \bigo{|Q|2^{2^{|Q|}}}$

\begin{lemma}
The size of $\csaiiof A$ is in  $\bigo{|Q|2^{2^{|Q|}}}$.
\end{lemma}

From the construction of $(\csaiiof{A}$ and the described encoding of configurations of $\scaut{A}$ by those of $\csaiiof{\caof R}$, it follows that 
$L(\csaiiof{\caof R}) = L(\scaut{\caof R})$.
Furthermore, for $R$ with flat counting, $\caof{R}$ is Cartesian, thus we get:

\begin{lemma}
  $L(\caof R) = L(\csaiiof {\caof R})$.
\end{lemma}

\Cref{sec:matching} gives worst-case time approximations for membership check with automata with non-replicating counters.
We now reduce these approximations for $\csaiiof{A}$.
To find the next CSA transition from state $(R,\ac)$ to be taken with a set interpretation $\smem$ and symbol $a$, we must traverse through all $a$-transitions starting in $(R,\ac)$ whose guards are minterms.
We need to test each minterm, however we know that testing of atomic predicate is constant time operation, therefore finding the next CSA transition can be done in $O(|Q|2^{2|Q|})$ as there are at most $O(2^{|Q|})$ minterms whose size is at most $O(|Q|2^{|Q|})$ (with similar reasoning as for the size of $\csaiiof A$).

The step that applies the update of selected transitions is amortized $O(|\ii X|) = \bigo{|X|2^{|Q|}}$, however, we can reduce this to $\bigo{2^{|Q|}}$ for flat counting.
The cost of entire matching of a word $w$ is then $\bigo{|w|(|Q|2^{2|Q|} + 2^{|Q|})} = \bigo{|w||Q|2^{2|Q|}}$.

\begin{lemma}
  The cost of entire matching of a word $w$ with $\csaiiof{A}$ is in $\bigo{|w||Q|2^{2|Q|}}$.
\end{lemma}

\section{Proof of \cref{theo:class}}
\label{app:delayed}
\classtheorem*

The theorem is shown using the two technical lemmas below
that capture a crucial relation between paths through $\caof R$ and $\csaiiof{\caof R}$.
We use here the notation from \cref{sec:caproperties}. 

\begin{lemma}
  \label{lemma:twopaths}
  Let there be a path of $\caof{R}$ over a word $w$ from its initial state to some state $q_i$ and let $\pi_1$, $\pi_2$ be two paths of $x$-fanout from the state $q$ over a word $v$ such that $\pi_1$ increments $x$ $k+1$-times ending in the state $r$, and $\pi_2$ increments $x$ $k$-times ending in the state $p$. 
  Then, if $\csaiiof{\caof{R}}$ exists, then there is a path over $wv$ in $\csaiiof{\caof{R}}$ that ends with an active counter $x_S$ such that $\marked{r} \in S$ and $p \in S$. 
\end{lemma}
\begin{proof}
  Assume that $\csaiiof{\caof{R}}$ exists.
  As $\csaiiof{\caof{R}}$ is created using extension of classical subset construction, the path over $w$ must lead to a state $(T, \ac)$ s.t. $q \in T$.
  Let $v = a_1\dots a_n$, $\pi_1[i]$ ($\pi_2[i]$), $1 \leq i \leq n$, be the path containing first $i$ transition of $\pi_1$ ($\pi_2$), and $r_i$ ($p_i$) be the state in which of the path $\pi_1[i]$ ($\pi_2[i]$) ends (i.e. $r_i = r$).
  Furthermore, let $\pi$ be the path from $(T, \ac)$ over $v$ in $\csaiiof{\caof{R}}$ and $(T_i, \ac_i)$ be the ending state of $\pi[i]$, for each $i$, $1 \leq i \leq n$.
  We now show that for each $i$, $1 \leq i \leq n$, the numbers of increments of $x$ occurring in $\pi_1[i]$ and $\pi_2[i]$ are either
  \begin{itemize}
    \item same and there is $x_{S_i} \in \ac_i$ where $r_i, p_i \in S_i$, or
    \item differ by one. Here, if $\pi_1[i]$ contains one more increment than $\pi_2[i]$, then there is $x_{S_i} \in \ac_i$ where $\marked{r_i}, p_i \in S_i$ (similarly if $\pi_2[i]$ contains one more increment).
  \end{itemize}
  We show this by induction on $i$.

  For $i = 1$, 
  because $\pi_1$ and $\pi_2$ are in $x$-fanout, their first transition from $q$ to $r_1$ and $p_1$ must assign $1$ to $x$.
  Therefore, the counters will be shared during the construction, i.e. there is $x_{S_1}$ to which we assign $1$ in the first transition of $\pi$ and $r_1, p_1 \in S_1$.
  
  Let now assume the statement holds for $1 \leq k < n$ and we want to show that it also holds for $k+1$.
  First, because $k+1 > 1$ and $\pi_1$ and $\pi_2$ are in $x$-fanout, we have that the transitions at the end of $\pi_1[k+1]$ and $\pi_2[k+1]$ must either copy or increment $x$.
  We now assume that the number of increments of $x$ occurring in $\pi_1[k]$ is one more than in $\pi_2[k]$ (the case when they are the same has similar reasoning).
  We know that there is $x_{S_k} \in \ac_k$ where $\marked{r_k}, p_k \in S_k$.
  If the transitions at the end of $\pi_1[k+1]$ and $\pi_2[k+1]$ have the same operation for $x$, then the same operation (possibly with filter) is done at the end of $\pi[k+1]$ for $x_{S_k}$ and the statement hold for $k+1$ (because the counters are still shared with the same increment postponing).
  If the transition at the end of $\pi_1[k+1]$ just copies $x$ while the one at the end of $\pi_2[k+1]$ increments it, then $x_{S_k}$ is incremented at the end of $\pi[k+1]$ (i.e. increments are synchronized), there is $x_{S_i} \in \ac_i$ where $r_i, p_i \in S_i$, and $\pi_1[k+1]$ and $\pi_2[k+1]$ have the same number of increments.
  Finally, the situation where the transition at the end of $\pi_1[k+1]$ increments $x$ while the one at the end of $\pi_2[k+1]$ copies it cannot happen.
  This is because the transition at the end of $\pi$ would have to do two increments on $x_{S_k}$ which is not allowed, i.e. $\csaiiof{\caof{R}}$  would not exist which is a contradiction.
\end{proof}
\begin{lemma}
  \label{lemma:twopaths2}
  Let during the construction of $\csaiiof{\caof{R}}$ be there a path $\pi$ over a word $w$ such that it ends in an active counter $x_S$ with some $\marked{r} \in S$.
  Then there are words $u,v$, $w = uv$ such that there is a path of $\caof{R}$ over $u$ going to some state $q$ and two paths $\pi_1$ and $\pi_2$ of $x$-fanout from the state $q$ over a word $v$ such that $\pi_1$ increments $x$ $k+1$-times ending in the state $r$, and $\pi_2$ increments $x$ $k$-times, for some $k \in \nat$.
\end{lemma}
\begin{proof}
  Let $w = a_1\dots a_n$ and $\move{(T_{i-1}, \ac_{i-1})}{a_i, \phi_i, u_i}{(T_i,\ac_i)}$, $1 \leq i \leq n$, be the transitions of $\pi$.
  Let $x_{S_l},x_{S_{l+1}}, \dots, x_{S_n}$, for some $1 \leq l \leq n$, be a sequence of counters where ${S_n} = {S}$, for each $l < j \leq n$ we have that term $u_j(x_{S_j})$ contains $x_{S_{j-1}}$, and $u(x_{S_l})$ does not contain any counter.
  Intuitively, this is a sequence of names of the same counter that starts by mapping either 0 or 1 to $x_{S_l}$ (and we will show it can actually be only 1) and then the values of this counter are further just copied or incremented to following counter until we end with $x_S$.
  Furthermore, $l \not= n$, otherwise $u_n(x_{S})$ would contain no counter and there would be no reason for $\marked{r}$ to be in $S$.
  It also holds that each $S_j \subseteq \states{x}$, as each $u_j(x_{S_j})$ contains $x_{S_{j-1}}$ and from the way the construction works, this means that there is a transition from each state of $S_{j-1}$ to some state of $S_{j}$ with $x$ being copied or incremented, i.e. they are transitions from $\trinner$ or $\trinc$ (which also means that $S_l \subseteq \states{x}$).
  We can also show that $S_l = \entry{x}$, because $u(x_{S_l})$ does not contain any counter, i.e. $u(x_{S_l}) = 1$ and this update is constructed from some transitions of $\trin$ or $\trrestart$ 
  We now take $u = a_1\dots a_{l-1}$, $v = a_l \dots a_n$ and the state $q$ is some initial state of the transitions from $\trin$ or $\trrestart$.
  Obviously, there is a path of $\caof{R}$ over $u$ going to $q$ and we can show by induction that for each $x_{S_j}$ we have that:
  \begin{itemize}
    \item if $q_1, q_2 \in S_j$, then there are paths $\pi_1$ to $q_1$ and $\pi_2$ to $q_2$ of $x$-fanout from the state $q$ over a word $a_l \dots a_j$ such that they both increment $x$ same amount of times and
    \item if $\marked{q_1}, q_2 \in S_j$, then there are paths $\pi_1$ to $q_1$ and $\pi_2$ to $q_2$ of $x$-fanout from the state $q$ over a word $a_l \dots a_j$ such that $\pi_1$ increments $x$ one more time than $\pi_2$.
  \end{itemize}
  The proof of this is very similar to the proof used for \cref{lemma:twopaths}. 
  Finally, $S$ must contain some non-marked state $p$, otherwise it would not be possible to have $\marked{r}$ in $S$ ($S$ cannot contain only marked states, because otherwise increment would have been synchronized in $u_n(c_{S})$).
\end{proof}


\begin{proof}[Proof of \cref{theo:class}]
%
\item{($\Rightarrow$)}
By contradiction.
Assume that $\caof{R}$ exists but $R$ is does not have synchronizing counting, hence, it has a sub-expression $S\regex{\{m,n\}}$ which does not have synchronizing counting.
There exist words $u = u_1,\ldots,u_k$ and $v = v_1,\ldots,v_{k+1}$ that are a counterexample to $S\regex{\{m,n\}}$ having synchronizing counting, i.e. $v$ is a prefix of $u$.  
We assume w.l.o.g. that $k$ is the smallest number with which such a counterexample can be obtained.
In $\caof R$, since it satisfies the conditions of \cref{sec:caproperties}, there is
a path $\pi$ to a source $q_i$ of an entry transition on the counting loop of
the counter $x$ that is connected with the sub-expression $S\regex{\{m,n\}}$.
Let $w$ be the word of that path. 
%
%
There must be two paths of $\caof R$ in $x$-fanout from $q_i$ over $v$. 
One of them, $\pi_1$, ends by a transition to some state $q$ in $\last{x}$ and it increments $x$ $k$-times (corresponding to processing
the word $v_1\cdots v_{k+1}$, by \cref{lemma:caofr}); the other one, $\pi_2$, corresponds to a prefix of $u_1\cdots u_{k}$. 
By \cref{lemma:caofr} we know it must increment $x$ at most $k-1$ times and if it incremented it fewer times, it would mean that we could remove some words $u_l\ldots u_k$, $l \leq k$, from $u$ and $v$ would still be a prefix of this word.
This would mean that we could have taken smaller $k$ and we would be still able to find counterexample to synchronizing counting, but $k$ is already smallest such one.
Therefore, $\pi_2$ increments $x$ exactly $k-1$ times and we can now apply \cref{lemma:twopaths} to find path of $\csaiiof{\caof{R}}$ over $wv$ such that it ends in state $(T, \ac)$ with some active counter $x_S \in \ac$ where $\marked{q} \in S$.
However, because $q \in \last{x}$, there is some transition $\delta$ in $\caof{R}$ that starts in the state $q$ and has $x \ass x+1$ in its update.
Now, during the construction of some transition from $(R, \ac)$ that uses $\delta$ we would get irresolvable counter replication and $\csaiiof{\caof{R}}$ would not exist, which is a contradiction.

\item{($\Leftarrow$)} By contradiction. Assume that $R$ has synchronizing counting, but $\csaiiof{\caof R}$ does not exist. 
Take the shortest path $\pi$ over some word $w$ to some state $(T, \ac)$ such that the construction of $\csaiiof{\caof{R}}$ failed while creating some transition from $(T, \ac)$.
This must mean that there is some $x_S \in \ac$ with some $\marked{q} \in S$ and $q \in \last{x}$.
Let $S\regex{\{m,n\}}$ be the sub-expression of $R$ whose counting loop in
$\caof R$ uses a counter $x$.
From \cref{lemma:twopaths2} we get two words $u,v$, $w = uv$ and two paths $\pi_1$ and $\pi_2$ in some $x$-fanout of some $q_i$ where $\pi_1$ ends in $q$ and increments $x$ $k+1$-times and $\pi_2$ increments $x$ $k$-times.
From \cref{lemma:caofr} and  $q \in \last{x}$, we can divide $v$ to $v_1,\ldots,v_{k+2}$ (corresponding to the path $\pi_1$) where $v_i \in L(S)$.
If we add to $\pi_2$ some path over some word $v'$ so that it ends in some state from $\last{x}$ (without adding any transition that increment $x$), we can then divide $vv'$ to $u_1\ldots u_{k+1}$.
However, $v$ is a prefix of $vv'$, which means that we found words $v_1,\ldots,v_{k+2}$ and $u_1\ldots u_{k+1}$ that show that $S\regex{\{m,n\}}$ does not have synchronizing counting, which is a contradiction.
\end{proof}

\section{Counterexample for Class From~\cite{oopsla}}
\label{app:counterexample}
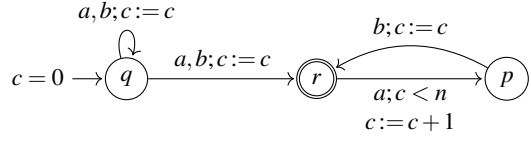
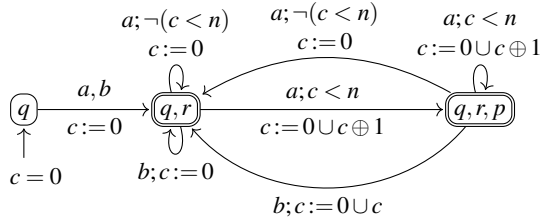
\begin{figure}
  \centering
  \begin{subfigure}{\columnwidth}
  \centering
  \begin{tikzpicture}[shorten >=1pt,node distance=2.5cm,on grid,initial text={$\phantom{\{\}}c=0$}]
    \node[state,initial]   (q_0)                {$q$};
    \node[state,accepting] (q_1) [right=of q_0] {$r$};
    \node[state]           (q_2) [right=of q_1] {$p$};
    \path[->] (q_0) edge                node [above] {$a,b; c \ass c$} (q_1)
                    edge [loop above]   node         {$a,b; c \ass c$} ()
              (q_1) edge                node [below,align=center] {$a;c < n$\\$c \ass c + 1$} (q_2)
              (q_2) edge [bend right]   node [above] {$b; c \ass c$} (q_1);
  \end{tikzpicture}
  \caption{CA created using the construction from~\cite{oopsla}}
  \label{fig:CAnonuniform}
  \end{subfigure}
  \begin{subfigure}{\columnwidth}
  \centering
  \begin{tikzpicture}[shorten >=1pt,node distance=4cm,on grid,initial text={$\phantom{\{\}}c=0$}]
    \node[mstate,initial below]   (q_0)                {$q$};
    \node[mstate,accepting
           ] (q_1) [right=2cm of q_0] {$q,r$};
    \node[mstate,accepting]           (q_2) [right=of q_1] {$q,r,p$};
    \path[->] (q_0) edge                node [above] {$a,b$} node [below] {$c \ass 0$} (q_1)
              (q_1) edge                node [above] {$a; c < n$} node [below] {$c \ass 0 \cup \finc{c}$} (q_2)
                    edge [loop above]   node [align=center] {$a; \neg (c < n)$\\$c \ass 0$} ()
                    edge [loop below]   node                {$b; c \ass 0$} ()
              (q_2) edge [bend right]   node [above,align=center] {$a; \neg (c < n)$\\$c \ass 0$} (q_1)
                    edge [bend left=50] node [below] {$b; c \ass 0 \cup c$} (q_1)
                    edge [loop above]   node [align=center] {$a; c < n$\\$c \ass 0 \cup \finc{c}$} ();
  \end{tikzpicture}
  \caption{Determinized CSA by algorithm from~\cite{oopsla}}
  \label{fig:CsAnonuniform}
  \end{subfigure}
  \caption{CA and determinized CSA for regex \regex{.*.(ab)\{n\}} where the final condition is $c \geq m$ for the denoted states and $\bot$ for all the other, and the missing guards are $\top$.}
\end{figure}

The authors of~\cite{oopsla} claimed that if regex $R$ contains counting loops only of type\\ \regex{$(\alpha_1\alpha_2\dots\alpha_n)\{n\}$} where $\alpha_1$ denotes a set of symbols disjoint with any symbols from $\alpha_i$, $2 \leq i \leq n$, then the created CA can be determinized by their algorithm into CSA that accepts the same language.
However, for regex \regex{.*.(ab)\{n\}}, which satisfies the syntactic criterion, we get a counting automaton $A$ on \cref{fig:CAnonuniform} created by construction from~\cite{oopsla}. 
The determinized CSA $A'$ shown on \cref{fig:CsAnonuniform} is created according to the algorithm from~\cite{oopsla}.
However, the word containing only $m+1$ symbols $a$ is accepted by $A'$ but not by $A$.

\section{Recognizing Letter Marked Counting}
\label{sec:checker}
%
The following simple procedure determines whether a regex $R$ has letter-marked counting. It inductively computes for a regex $R$ the set $T_R$ of all sets of marker-letters. We use $\Sigma(R)$ as the set of all symbols occurring in $R$.
\begin{itemize}
\item For $R = \regex a, a\in \Sigma$, we have $T_R = \set{\set{a}}$, as there is only one word $a \in L(R)$. For $R = \epsilon$, we have $T_{R} = \{\}$.
\item For $R = \rexuni{R_1}{R_2}$, we know that for each $T_1 \in T_{R_1}$, each word in $L(R_1)$ contains exactly one symbol from $T_1$. However, there might be words in $L(R_2)$ that do not contain any symbol from $T_1$ or they contain more than one symbol. If we are given some $T_2$ from $T_{R_2}$, we can see that symbols from $\Sigma(R_2) \setminus T_2$ can occur more than once, or not occur at all, therefore, we can just test if $ T_1 \cap (\Sigma(R_2) \setminus T_2) = \emptyset$. We then take \[T_R = \{\,T_1 \cup T_2\,|\,T_1 \in T_{R_1}, T_2 \in T_{R_2}, T_1 \cap (\Sigma(R_2) \setminus T_2) = \emptyset \text{ and } T_2 \cap (\Sigma(R_1) \setminus T_1) = \emptyset\,\}.\]
\item For $R = \rexconcat{R_1}{R_2}$, we know that symbols from $T \in T_{R_1}$ occur in words from $\rexconcat{R_1}{R_2}$ exactly once, if they do not occur in the word given by $R_2$, i.e. if $T \cap \Sigma(R_2) = \emptyset$. We can therefore take \[T_R = \{\,T\,|\,T \in T_{R_1} \text{ and } T \cap \Sigma(R_2) = \emptyset\,\} \cup \{\,T\,|\,T \in T_{R_2} \text{ and } T \cap \Sigma(R_1) = \emptyset\,\}.\]
\item For $R = \rexstar{R'}$, no set of marker-letters exists, as repeating words from $R'$ forces each letter to occur more than once. For this case, $T_{R} = \{\}$.
\end{itemize}

\section{Regexes with Non-Synchronizing Counting}
\label{sec:nonvisible}
The following is a list of 24 non-synchronizing flat regexes with the sum of upper bounds of counters larger than 20 from the benchmark of \cite{oopsla}:
\begin{lstlisting}
^(.*?){1,128}$
(.*){1,32000}[bc]
^(.*){0,254}$
(.+){25}(.*)
((?:[^\n]*\n?){1,40})
(\n\s+)(criterion .*\n)(\s.+){1,99}
^(([\w\d\-_]+)\W([\w\d]+)\W){1,32}? *(.+)
^(([\w\d\-_]+)\W([\w\d]+)\W){1,32}? *(\w+.+)
^[a-z0-9]+([._\\-]*[a-z0-9])*@([a-z0-9]+[-a-z0-9]*[a-z0-9]+.){1,63}[a-z0-9]+$
^[a-z0-9]+([._\\-]*[a-z0-9])*@(\w+([-.]\w+)*\.){1,63}[a-z0-9]+$
^(([0-9]+\s*){1,255})(.*)?$
REK\: ([a-zA-Z]{2}[0-9]{2}[a-zA-Z0-9]{4}[0-9]{7}([a-zA-Z0-9]?){0,16})
ICE_Dims.{92}((_?(X|\d+)){13})
[a-zA-Z]{2}[0-9]{2}[a-zA-Z0-9]{4}[0-9]{7}([a-zA-Z0-9]?){0,16}
https?://(?:\S+/){4}([0-9a-f]{40})/?([^#\s]+)?(?:#(\S+))?
/\r\n\s*Accept-Language\s*|3a|\s*([^\r\n]*?\x2c){20}/mi
^jdbc:db2://((?:(?:(?:25[0-5]|2[0-4][0-9]|[01]?[0-9][0-9]?).){3}(?:25[0-5]|2[0-4][0-9]|[01]?[0-9][0-9]?))|(?:(?:(?:(?:[A-Z|a-z])(?:[\w|-]){0,61}(?:[\w]?[.]))*)(?:(?:[A-Z|a-z])(?:[\w|-]){0,61}(?:[\w]?)))):([0-9]{1,5})/([0-9|A-Z|a-z|_|#|$]{1,16})$
^[-\w&#39;+*$^&%=~!?{}#|/`]{1}([-\w&#39;+*$^&%=~!?{}#|`.]?[-\w&#39;+*$^&%=~!?{}#|`]{1}){0,31}[-\w&#39;+*$^&%=~!?{}#|`]?@(([a-zA-Z0-9]{1}([-a-zA-Z0-9]?[a-zA-Z0-9]{1}){0,31})\.{1})+([a-zA-Z]{2}|[a-zA-Z]{3}|[a-zA-Z]{4}|[a-zA-Z]{6}){1}$
^([a-z0-9]+([\-a-z0-9]*[a-z0-9]+)?\.){0,}([a-z0-9]+([\-a-z0-9]*[a-z0-9]+)?){1,63}(\.[a-z0-9]{2,7})+$
[-\w'+*$^&%=~!?{}#|/`]{1}([-\w'+*$^&%=~!?{}#|`.]?[-\w'+*$^&%=~!?{}#|`]{1}){0,31}[-\w'+*$^&%=~!?{}#|`]?@(([a-zA-Z0-9]{1}([-a-zA-Z0-9]?[a-zA-Z0-9]{1}){0,31})\.{1})+([a-zA-Z]{2}|[a-zA-Z]{3}|[a-zA-Z]{4}|[a-zA-Z]{6}){1}
jdbc:db2://((?:(?:(?:25[0-5]|2[0-4][0-9]|[01]?[0-9][0-9]?).){3}(?:25[0-5]|2[0-4][0-9]|[01]?[0-9][0-9]?))|(?:(?:(?:(?:[A-Z|a-z])(?:[\w|-]){0,61}(?:[\w]?[.]))*)(?:(?:[A-Z|a-z])(?:[\w|-]){0,61}(?:[\w]?)))):([0-9]{1,5})/([0-9|A-Z|a-z|_|#|$]{1,16})
/^\s*Accept-Language\s*|3a|\s*([^\r\n]*?\x2c){20}/mi
/PUTOLF\s+((\S+\s+){4}[^\s]{256}|(\S+\s+){6}[^\x3c]{512})/i
/^.*HTTP.*\r\n(.+\x3a\s+.+\r\n){31,}/
\end{lstlisting}

\end{document}